\documentclass[english,final]{IEEEtran}
\usepackage[T1]{fontenc}
\usepackage[latin9]{inputenc}
\usepackage{amsthm}
\usepackage{amsmath}
\usepackage{amssymb}
\usepackage{graphicx}
\usepackage{esint}
\usepackage{flushend}
\usepackage{cite} 
\usepackage{enumerate}
\usepackage{epstopdf} 
 
\makeatletter

\newcommand{\openone}{\leavevmode\hbox{\small1\normalsize\kern-.33em1}}

\newcommand{\pebar}{\overline{p}_{e}}
\newcommand{\pe}{p_{e}} 

\newcommand{\Eziid}{E_{0}^{\mathrm{iid}}}
\newcommand{\Ezcc}{E_{0}^{\mathrm{cc}}}
\newcommand{\Eohatcc}{\hat{E}_{0}^{\mathrm{cc}}}
\newcommand{\Ezcost}{E_{0}^{\mathrm{cost}}}
\newcommand{\Ezcostprime}{E_{0}^{\mathrm{cost}^{\prime}}}
\newcommand{\Eriid}{E_{r}^{\mathrm{iid}}}
\newcommand{\Ercc}{E_{r}^{\mathrm{cc}}}

\newcommand{\Ercost}{E_{r}^{\mathrm{cost}}}
\newcommand{\Ercostprime}{E_{r}^{\mathrm{cost}^{\prime}}}

\newcommand{\rcu}{\mathrm{rcu}}
\newcommand{\rcus}{\rcu_{s}}
\newcommand{\rcuss}{\rcu_{s}^{*}}
\newcommand{\rcushat}{\widehat{\mathrm{rcu}}_{s}}
\newcommand{\rcusshat}{\widehat{\mathrm{rcu}}_{s}^{*}}

\newcommand{\rhohat}{\hat{\rho}}
\newcommand{\LM}{I_{\mathrm{LM}}}
\newcommand{\GMI}{I_{\mathrm{GMI}}}

\newcommand{\SetScost}{\mathcal{S}}

\newcommand{\SetTcost}{\mathcal{T}}

\newcommand{\Miid}{M^{\mathrm{iid}}}
\newcommand{\Mcc}{M^{\mathrm{cc}}}
\newcommand{\Mcost}{M^{\mathrm{cost}}}
\newcommand{\Rcr}{R^{\mathrm{cr}}}
\newcommand{\Rcrs}{R_s^{\mathrm{cr}}}

\newcommand{\qbar}{\overline{q}}
\newcommand{\rbar}{\overline{r}}

\newcommand{\Ptilde}{\widetilde{P}}

\newcommand{\hover}{\overline{h}}

\newcommand{\xbar}{\overline{x}}
\newcommand{\Xbar}{\overline{X}}
\newcommand{\xtilde}{\widetilde{x}}

\newcommand{\xvbar}{\overline{\boldsymbol{x}}}

\newcommand{\xv}{\boldsymbol{x}}

\newcommand{\Xvbar}{\overline{\boldsymbol{X}}}

\newcommand{\Xv}{\boldsymbol{X}}
\newcommand{\yv}{\boldsymbol{y}}
\newcommand{\Yv}{\boldsymbol{Y}}

\newcommand{\Ac}{\mathcal{A}}
\newcommand{\Cc}{\mathcal{C}}
\newcommand{\Dc}{\mathcal{D}}

\newcommand{\Fc}{\mathcal{F}}
\newcommand{\Gc}{\mathcal{G}}
\newcommand{\Ic}{\mathcal{I}}

\newcommand{\Pc}{\mathcal{P}}
\newcommand{\Sc}{\mathcal{S}}

\newcommand{\Xc}{\mathcal{X}}
\newcommand{\Yc}{\mathcal{Y}}

\newcommand{\EE}{\mathbb{E}}
\newcommand{\RR}{\mathbb{R}}
\newcommand{\PP}{\mathbb{P}}
\newcommand{\ZZ}{\mathbb{Z}}
\newcommand{\Qsf}{\mathsf{Q}}

\newcommand{\defeq}{\triangleq}
\newcommand{\var}{\mathrm{Var}}

\newcommand{\alphanl}{\alpha^{\mathrm{nl}}_n}
\newcommand{\alphal}{\alpha^{\mathrm{l}}_n}
\newcommand{\betanl}{\beta^{\mathrm{nl}}_n}
\newcommand{\betal}{\beta^{\mathrm{l}}_n} 
  
\theoremstyle{plain}
\newtheorem{prop}{\protect\propositionname} 
\theoremstyle{plain}
\newtheorem{thm}{\protect\theoremname}
\theoremstyle{plain}
\newtheorem{lem}{\protect\lemmaname}

\DeclareMathOperator*{\argmax}{arg\,max}

\makeatother
 
\usepackage{babel}
  \providecommand{\propositionname}{Proposition}
\providecommand{\theoremname}{Theorem}
\providecommand{\lemmaname}{Lemma}

\begin{document}

\title{\fontsize{23.5}{30} \selectfont Mismatched Decoding: Error Exponents, Second-Order Rates and Saddlepoint Approximations}

\author{Jonathan Scarlett, Alfonso Martinez, \IEEEmembership{Senior Member, IEEE}, and Albert Guill\'en i F\`abregas, \IEEEmembership{Senior Member, IEEE}}

\maketitle
\long\def\symbolfootnote[#1]#2{\begingroup\def\thefootnote{\fnsymbol{footnote}}\footnote[#1]{#2}\endgroup}
\begin{abstract}
    This paper considers the problem of channel coding with a given (possibly
    suboptimal) maximum-metric decoding rule.  A cost-constrained random-coding 
    ensemble with multiple auxiliary costs is introduced, and is shown to
    achieve error exponents and second-order coding rates matching those
    of constant-composition random coding, while being directly applicable 
    to channels with infinite or continuous alphabets.  The number of auxiliary
    costs required to match the error exponents and second-order rates of constant-composition 
    coding is studied, and is shown to be at most two.
    For i.i.d. random coding, asymptotic estimates of two well-known 
    non-asymptotic bounds are given using saddlepoint approximations.
    Each expression is shown to characterize the asymptotic behavior of
    the corresponding random-coding bound at both
    fixed and varying rates, thus unifying the regimes characterized
    by error exponents, second-order rates and moderate deviations.  For fixed rates, novel
    exact asymptotics expressions are obtained to within a multiplicative
    $1+o(1)$ term.  Using numerical examples, it is shown that the 
    saddlepoint approximations are highly accurate even at short block lengths.  
\end{abstract}
\symbolfootnote[0]{
    J. Scarlett is with the Department of Engineering, University of Cambridge, 
    Cambridge, CB2 1PZ, U.K. (e-mail: jmscarlett@gmail.com).  A. Martinez is 
    with the Department of Information and Communication Technologies, 
    Universitat Pompeu Fabra, 08018 Barcelona, Spain (e-mail: alfonso.martinez@ieee.org).  
    A. Guill\'en i F\`abregas is with the Instituci\'o Catalana de Recerca i Estudis 
    Avan\c{c}ats (ICREA), the Department of Information and Communication Technologies, 
    Universitat Pompeu Fabra, 08018 Barcelona, Spain, and also with the Department of 
    Engineering, University of Cambridge, Cambridge, CB2 1PZ, U.K. (e-mail: guillen@ieee.org).
    This work has been funded in part by the European Research Council under 
    ERC grant agreement 259663, by the European Union's 7th Framework Programme 
    under grant agreement 303633 and by the Spanish Ministry of 
    Economy and Competitiveness under grants RYC-2011-08150 and TEC2012-38800-C03-03.  
    This work was presented in part at the Allerton Conference on Communication, 
    Computing and Control (2012), and at the Information Theory and Applications 
    Workshop (2013, 2014).}

\section{Introduction} \label{sec:INTRO}

Information-theoretic studies of channel coding typically seek to characterize
the performance of coded communication systems when the encoder and decoder
can be optimized.  In practice, however, optimal decoding rules are 
often ruled out due to channel uncertainty and implementation
constraints. In this paper, we consider the mismatched decoding problem
\cite{Csiszar1,Hui,Compound,Csiszar2,Merhav,ConverseMM,MMRevisited,MacMM},
in which the decoder employs maximum-metric decoding with a metric
which may differ from the optimal choice.

The problem of finding the highest achievable rate possible
with mismatched decoding is open, and is generally believed to be difficult.
Most existing work has focused on achievable rates via random coding; see
Section \ref{sub:CONTRIBUTIONS} for an outline.
The goal of this paper is to present a more comprehensive analysis
of the random-coding error probability under various ensembles,
including error exponents \cite[Ch. 5]{Gallager}, second-order 
coding rates \cite{Strassen,Finite,Hayashi}, and refined
asymptotic results based on the saddlepoint approximation \cite{SaddlepointBook}.

\subsection{System Setup} \label{sub:SU_SYSTEM_SETUP}

The input and output alphabets are denoted by $\Xc$ and $\Yc$ respectively.
The conditional probability of receiving an output vector $\yv=(y_1,\cdots,y_n)$ given
an input vector $\xv=(x_1,\cdots,x_n)$ is given by 
\begin{equation}
    W^{n}(\yv|\xv)\triangleq\prod_{i=1}^{n}W(y_{i}|x_{i}) \label{eq:SU_DMC}
\end{equation}
for some transition law $W(y|x)$.  Except where
stated otherwise, we assume that $\Xc$ and $\Yc$
are finite, and thus the channel is a discrete memoryless channel
(DMC). The encoder takes as input
a message $m$ uniformly distributed on the set $\{1,\dotsc,M\}$,
and transmits the corresponding codeword $\xv^{(m)}$ from
a codebook $\Cc=\{\xv^{(1)},\dotsc,\xv^{(M)}\}$.
The decoder receives the vector $\yv$ at the output of
the channel, and forms the estimate 
\begin{equation}
    \hat{m}=\argmax_{j\in\{1,\dotsc,M\}} q^{n}(\xv^{(j)},\yv),\label{eq:SU_DecodingRule}
\end{equation}
where $q^{n}(\xv,\yv)\triangleq\prod_{i=1}^{n}q(x_{i},y_{i})$.
The function $q(x,y)$ is assumed to be non-negative, and is 
called the \emph{decoding metric}.  In the case of a tie,
a codeword achieving the maximum in \eqref{eq:SU_DecodingRule} is selected uniformly at random.
It should be noted that maximum-likelihood (ML) decoding is a special 
case of \eqref{eq:SU_DecodingRule}, since it is recovered by setting $q(x,y)=W(y|x)$.

An error is said to have occurred if $\hat{m}$ differs from $m$.
A rate $R$ is said to be achievable if, for all $\delta>0$, there
exists a sequence of codes $\Cc_n$ of length $n$ with $M \ge e^{n(R-\delta)}$ and
vanishing error probability $p_e(\Cc_n)$.  An error exponent $E(R)$ is said to
be achievable if there exists a sequence of codebooks $\Cc_{n}$
of length $n$ and rate $R$ such that
\begin{equation}
    \liminf_{n\to\infty}-\frac{1}{n}\log \pe(\Cc_{n})\ge E(R).
\end{equation}
We let $\pebar(n,M)$ denote the average error probability
with respect to a given random-coding ensemble which will be clear
from the context. The random-coding error exponent $E_{r}(R)$ is
said to exhibit ensemble tightness if
\begin{equation}
    \lim_{n\to\infty}-\frac{1}{n}\log\pebar(n,e^{nR})=E_{r}(R).\label{eq:SU_EnsTight}
\end{equation}

\subsection{Notation}

The set of all probability distributions on an alphabet $\Xc$
is denoted by $\Pc(\Xc)$, and the set of all empirical
distributions on a vector in $\Xc^{n}$ (i.e. types \cite[Sec. 2]{CsiszarBook}\cite{GallagerCC})
is denoted by $\Pc_{n}(\Xc)$. The type of a 
vector $\xv$ is denoted by $\hat{P}_{\xv}(\cdot)$.
For a given $Q\in\Pc_{n}(\Xc)$, the type class $T^{n}(Q)$
is defined to be the set of all sequences in $\Xc^{n}$ with type $Q$. 

The probability of an event is denoted by $\PP[\cdot]$, and
the symbol $\sim$ means ``distributed as''. The marginals of a
joint  distribution $P_{XY}(x,y)$ are denoted by $P_{X}(x)$ and
$P_{Y}(y)$. We write $P_{X}=\Ptilde_{X}$ to denote element-wise
equality between two probability distributions on the same alphabet.
Expectation with respect to a joint distribution $P_{XY}(x,y)$ is
denoted by $\EE_{P}[\cdot]$, or $\EE[\cdot]$ when the associated probability
distribution is understood from the context. Similar notations
$I_{P}(X;Y)$ and $I(X;Y)$ are used for the mutual information. Given a distribution
$Q(x)$ and conditional distribution $W(y|x)$, we write $Q\times W$
to denote the joint distribution $Q(x)W(y|x)$.

For two positive sequences $f_{n}$ and $g_{n}$, we write $f_{n}\doteq g_{n}$
if $\lim_{n\to\infty}\frac{1}{n}\log\frac{f_{n}}{g_{n}}=0$, and we write $f_{n}\,\dot{\le}\,g_{n}$
if $\limsup_{n\to\infty}\frac{1}{n}\log\frac{f_{n}}{g_{n}}\le0$ and analogously for $\dot{\ge}$.
We write $f_n \asymp g_n$ if $\lim_{n\to\infty}\frac{f_n}{g_n}=1$,
and we make use of the standard asymptotic notations $O(\cdot)$, $o(\cdot)$, 
$\Theta(\cdot)$, $\Omega(\cdot)$ and $\omega(\cdot)$. 

We denote the tail probability of a zero-mean unit-variance Gaussian
variable by $\Qsf(\cdot)$, and we denote its functional inverse
by $\Qsf^{-1}(\cdot)$. 
All logarithms have base $e$, and all rates are in units of nats
except in the examples, where bits are used. We define $[c]^{+}=\max\{0,c\}$,
and denote the indicator function by $\openone\{\cdot\}$.

\subsection{Overview of Achievable Rates} \label{sub:CONTRIBUTIONS}

Achievable rates for mismatched decoding have been derived using the 
following random-coding ensembles:
\begin{enumerate}
    \item the i.i.d. ensemble, in which each symbol of each codeword
    is generated independently;
    \item the constant-composition ensemble, in which each codeword is drawn
    uniformly from the set of sequences with a given empirical distribution;
    \item the cost-constrained ensemble, in which each codeword is drawn
    according to an i.i.d. distribution conditioned on an auxiliary 
    cost constraint being satisfied.
\end{enumerate}
While these ensembles all yield the same achievable rate under ML decoding,
i.e. the mutual information, this is not true under mismatched decoding.

The most notable early works on mismatched decoding are by Hui \cite{Hui}
and Csisz\'{a}r and K\"{o}rner \cite{Csiszar1}, who used constant-composition
random coding  to derive the following achievable rate for mismatched DMCs, 
commonly known as the LM rate:
\begin{equation}
    \LM(Q)=\min_{\Ptilde_{XY}}I_{\Ptilde}(X;Y),\label{eq:INTR_PrimalLM}
\end{equation}
where the minimization is over all joint distributions satisfying 
\begin{align}
    \Ptilde_{X}(x)             & =Q(x)\label{eq:INTR_ConstrLM1} \\
    \Ptilde_{Y}(y)             & =\sum_{x}Q(x)W(y|x)\label{eq:INTR_ConstrLM2} \\
    \EE_{\Ptilde}[\log q(x,y)] & \ge \EE_{Q\times W}[\log q(x,y)].\label{eq:INTR_ConstrLM3} 
\end{align}
This rate can equivalently be expressed as \cite{MMRevisited}
\begin{equation}
    \LM(Q)\triangleq\sup_{s\ge0,a(\cdot)}\EE\left[\log\frac{q(X,Y)^{s}e^{a(X)}}{\EE[q(\Xbar,Y)^{s}e^{a(\Xbar)}\,|\, Y]}\right],\label{eq:INTR_RateLM}
\end{equation}
where $(X,Y,\Xbar) \sim Q(x)W(y|x)Q(\xbar)$.

Another well-known rate in the literature is the generalized mutual information 
(GMI) \cite{Compound,MMRevisited}, given by 
\begin{equation}
    \GMI(Q)=\min_{\Ptilde_{XY}}D\big(\Ptilde_{XY}\|Q\times\Ptilde_{Y}\big),\label{eq:INTR_PrimalGMI}
\end{equation}
where the minimization is over all joint distributions
satisfying \eqref{eq:INTR_ConstrLM2} and \eqref{eq:INTR_ConstrLM3}.
This rate can equivalently be expressed as
\begin{equation}
    \GMI(Q)\triangleq\sup_{s\ge0}\EE\left[\log\frac{q(X,Y)^{s}}{\EE[q(\Xbar,Y)^{s}\,|\, Y]}\right]. \label{eq:INTR_RateGMI}
\end{equation}
Both \eqref{eq:INTR_PrimalGMI} and \eqref{eq:INTR_RateGMI} can be derived
using i.i.d. random coding, but only the latter has been shown to remain
valid in the case of continuous alphabets \cite{Compound}.

The GMI cannot exceed the LM rate, and the latter can be strictly
higher even after the optimization of $Q$.  Motivated by this fact, 
Ganti \emph{et al.} \cite{MMRevisited} proved that \eqref{eq:INTR_RateLM}
is achievable in the case of general alphabets. This was done
by generating a number of codewords according to an i.i.d. distribution $Q$,
and then discarding all of the codewords for which $\big|\frac{1}{n}\sum_{i=1}^{n}a(x_{i})-\EE_{Q}[a(X)]\big|$
exceeds some threshold.  An alternative proof is given in \cite{Variations}
using cost-constrained random coding.

In the terminology of \cite{MMRevisited}, \eqref{eq:INTR_PrimalLM}
and \eqref{eq:INTR_PrimalGMI} are primal expressions, and \eqref{eq:INTR_RateLM}
and \eqref{eq:INTR_RateGMI} are the corresponding dual expressions.
Indeed, the latter can be derived from the former using Lagrange duality
techniques \cite{Convex,Merhav}.

For binary-input DMCs, a matching converse to the LM rate was reported by 
Balakirsky \cite{ConverseMM}.  However, in the general case, several examples
have been given in which the rate is strictly smaller than the mismatched capacity \cite{Csiszar2,Merhav,MacMM}.
In particular, Lapidoth \cite{MacMM} gave an improved rate using multiple-access
coding techniques.  See \cite{MMSomekh,JournalMU} for more recent studies
on the benefit of multiuser coding techniques, \cite{PaperExpurg} for a study
of expurgated exponents, and \cite{MMGeneralFormula} for multi-letter converse results. 

\subsection{Contributions}

Motivated by the fact that most existing work on mismatched decoding
has focused on achievable rates, the main goal of this paper is to
present a more detailed analysis of the random-coding error probability.
Our main contributions are as follows:
\begin{enumerate}
    \item In Section \ref{sec:RANDOM_CODING_ENSEMBLES}, we present a generalization
    of the cost-constrained ensemble in \cite[Ch 7.3]{Gallager}, \cite{Variations} to
    include multiple auxiliary costs.  This ensemble serves as an alternative to 
    constant-composition codes for improving the performance compared to i.i.d. codes, 
    while being applicable to channels with infinite or continuous alphabets.
    \item In Section \ref{sec:SU_ERROR_EXPONENT}, an ensemble-tight error 
    exponent is given for the cost-constrained ensemble.
    It is shown that the exponent for the constant-composition ensemble \cite{Csiszar1}
    can be recovered using at most two auxiliary costs, and sometimes fewer.
    \item In Section \ref{sec:SU_SECOND_ORDER_RATES}, an achievable second-order 
    coding rate is given for the cost-constrained ensemble.  
    Once again, it is shown that the performance of constant-composition
    coding can be matched using at most two auxiliary costs, and sometimes fewer.  Our
    techniques are shown to provide a simple method for obtaining 
    second-order achievability results for continuous channels.
    \item In Section \ref{sec:SU_REFINED_IID}, we provide refined asymptotic
    results for i.i.d. random coding.  For two non-asymptotic random-coding bounds 
    introduced in Section \ref{sec:RANDOM_CODING_ENSEMBLES}, we give
    saddlepoint approximations \cite{SaddlepointBook} that can be computed 
    efficiently, and that characterize the asymptotic behavior of 
    the corresponding bounds as $n\to\infty$ at all positive rates (possibly varying with $n$).
    In the case of fixed rates, the approximations recover the prefactor growth
    rates obtained by Altu\u{g} and Wagner \cite{RefinementJournal}, along with 
    a novel characterization of the multiplicative $O(1)$ terms.
    Using numerical examples, it is shown that the approximations are remarkably 
    accurate even at small block lengths.
\end{enumerate}

\section{Random-Coding Bounds and Ensembles} \label{sec:RANDOM_CODING_ENSEMBLES}

Throughout the paper, we consider random coding in which each codeword $\Xv^{(i)}$
($i=1,\cdots,M$) is independently generated according to a given distribution $P_{\Xv}$.
We will frequently make use of the following theorem, which provides variations
of the random-coding union (RCU) bound given by Polyanskiy \emph{et al.} \cite{Finite}.
\begin{thm} \label{thm:SU_FiniteUB}
    For any codeword distribution $P_{\Xv}(\xv)$ and constant $s\ge0$,
    the random-coding error probability $\pebar$ satisfies
    \begin{equation}
        \frac{1}{4}\rcu(n,M) \le \pebar(n,M) \le \rcu(n,M) \le \rcus(n,M),\label{eq:SU_UpperPe} 
    \end{equation}
    where
    \begin{align}
        \rcu(n,M)        &\triangleq \EE\big[\min\big\{1, \nonumber 
            \\ & \hspace*{-3mm} (M-1)\PP[q^{n}(\Xvbar,\Yv)\ge q^{n}(\Xv,\Yv)\,|\,\Xv,\Yv]\big\}\big] \label{eq:SU_RCU} \\
        \rcus(n,M)       &\triangleq \EE\bigg[\min\Big\{1,(M-1)\frac{\EE[q^{n}(\Xvbar,\Yv)^{s}\,|\,\Yv]}{q^{n}(\Xv,\Yv)^{s}}\Big\}\bigg] \label{eq:SU_RCU_s}
    \end{align}
with $(\Xv,\Yv,\Xvbar) \sim P_{\Xv}(\xv)W^n(\yv|\xv)P_{\Xv}(\xvbar)$.
\end{thm}
\begin{IEEEproof}
    Similarly to \cite{Finite}, we obtain the upper bound $\rcu$ by writing
    \begin{align}
        \pebar(n,M) &\le \PP\bigg[\bigcup_{i\ne m}\big\{ q^{n}(\Xv^{(i)},\Yv)\ge q^{n}(\Xv,\Yv)\big\}\bigg] \label{eq:SU_RCU1} \\
                    &= \EE\Bigg[\PP\bigg[\bigcup_{i\ne m}\big\{ q^{n}(\Xv^{(i)},\Yv)\ge q^{n}(\Xv,\Yv)\big\}\,\Big|\,\Xv,\Yv\bigg]\Bigg] \\
                    &\le \rcu(n,M), \label{eq:SU_RCU3}
    \end{align}
    where \eqref{eq:SU_RCU1} follows by upper bounding the random-coding error probability
    by that of the decoder which breaks ties as errors, and \eqref{eq:SU_RCU3} follows
    by applying the truncated union bound.  To prove the lower bound in \eqref{eq:SU_UpperPe}, 
    it suffices to show that each of the upper bounds in \eqref{eq:SU_RCU1} and \eqref{eq:SU_RCU3} is tight
    to within a factor of two.  The matching lower bound to \eqref{eq:SU_RCU1} follows
    since whenever a tie occurs it must be between at least two codewords \cite{TwoChannels},
    and the matching lower bound to \eqref{eq:SU_RCU3} follows since the union 
    is over independent events \cite[Lemma A.2]{ShulmanThesis}.  
    We obtain the upper bound $\rcus$ by applying Markov's inequality to the
    inner probability in \eqref{eq:SU_RCU}.
\end{IEEEproof}

In this paper, we consider the cost-constrained ensemble 
characterized by the following codeword distribution: 
\begin{equation}
    P_{\Xv}(\xv)=\frac{1}{\mu_{n}}\prod_{i=1}^{n}Q(x_{i})\openone\big\{\xv\in\Dc_{n}\big\},\label{eq:SU_Q_Multi}
\end{equation}
where
\begin{equation}
    \Dc_{n}\triangleq\bigg\{\xv\,:\,\bigg|\frac{1}{n}\sum_{i=1}^{n}a_{l}(x_{i})-\phi_{l}\bigg|\le\frac{\delta}{n},\, l=1,\dotsc,L\bigg\},\label{eq:SU_SetDn}
\end{equation}
and where $\mu_{n}$ is a normalizing constant, $\delta$ is a positive
constant, and for each $l=1,\dotsc,L$, $a_{l}(\cdot)$ is a 
real-valued function on $\Xc$, and $\phi_{l}\triangleq\EE_{Q}[a_{l}(X)]$.
We refer to each function $a_l(\cdot)$ as an auxiliary cost function, or simply a cost.
Roughly speaking, each codeword is generated according to an i.i.d.
distribution conditioned on the empirical mean of each cost function
$a_{l}(x)$ being close to the true mean. This generalizes the ensemble
studied in \cite[Sec. 7.3]{Gallager}, \cite{Variations} by including
multiple costs.

The cost functions $\{a_{l}(\cdot)\}_{l=1}^{L}$ in \eqref{eq:SU_Q_Multi} should not
be viewed as being chosen to meet a system constraint (e.g. power
limitations). Rather, they are introduced in order to improve the
performance of the random-coding ensemble itself. 
However, system costs can be handled similarly; see Section \ref{sec:CONCLUSION}
for details.  The constant $\delta$ in \eqref{eq:SU_Q_Multi} could, in principle,
vary with $l$ and $n$, but a fixed value will suffice for our purposes. 

In the case that $L=0$, it should be understood that $\Dc_n$ contains all $\xv$ 
sequences.  In this case, \eqref{eq:SU_Q_Multi} reduces to the i.i.d. ensemble, 
which is characterized by
\begin{equation} 
    P_{\Xv}(\xv)=\prod_{i=1}^{n}Q(x_{i}).\label{eq:SU_Q_IID}
\end{equation}
A less obvious special case of \eqref{eq:SU_Q_Multi} is the constant-composition
ensemble, which is characterized by
\begin{equation}
    P_{\Xv}(\xv)=\frac{1}{|T^{n}(Q_{n})|}\openone\big\{\xv\in T^{n}(Q_{n})\big\},\label{eq:SU_Q_CC}
\end{equation}
where $Q_{n}$ is a type such that $\max_{x}|Q_{n}(x)-Q(x)|\le\frac{1}{n}$.
That is, each codeword is generated uniformly over the type class
$T^{n}(Q_{n})$, and hence each codeword has the same composition.
To recover this ensemble from \eqref{eq:SU_Q_Multi}, we replace
$Q$ by $Q_n$ and choose the parameters $L=|\Xc|$, $\delta<1$ and  
\begin{equation}
    a_{l}(x)=\openone\{x=l\}, \quad l=1,\cdots,|\Xc|, \label{eq:SU_CostCaseCC}
\end{equation}
where we assume without loss of generality that $\Xc=\{1,\cdots,|\Xc|\}$.

The following proposition shows that the normalizing constant $\mu_{n}$
in \eqref{eq:SU_Q_Multi} decays at most polynomially in $n$. When
$|\Xc|$ is finite, this can easily be shown using the method
of types. In particular, choosing the functions given in the previous
paragraph to recover the constant-composition ensemble, we have $\mu_{n}\ge(n+1)^{-(|\Xc|-1)}$
\cite[p. 17]{CsiszarBook}. For the sake of generality, we present
a proof which applies to more general alphabets, subject to minor
technical conditions. The case $L=1$ was handled in \cite[Ch. 7.3]{Gallager}.
\begin{prop} \label{prop:SU_SubExpCost}
    Fix an input alphabet $\Xc$ (possibly
    infinite or continuous), an input distribution $Q\in\Pc(\Xc)$
    and the auxiliary cost functions $a_{1}(\cdot),\cdots,a_{L}(\cdot)$. If
    $\EE_Q[a_{l}(X)^2]<\infty$ for $l=1,\dotsc,L$,
    then there exists a choice of $\delta>0$ such that the normalizing constant
    in \eqref{eq:SU_Q_Multi} satisfies $\mu_{n}=\Omega(n^{-L/2})$.
\end{prop}
\begin{IEEEproof}
    This result follows from the multivariate local limit theorem in \cite[Cor. 1]{Stone}, 
    which gives asymptotic expressions for probabilities of i.i.d. random 
    vectors taking values in sets of the form 
    \eqref{eq:SU_SetDn}.  Let $\Sigma$ denote the covariance matrix of
    the vector $[a_{1}(X),\dotsc,a_{L}(X)]^{T}$.  We have by assumption that 
    the entries of $\Sigma$ are finite.  Under the additional assumption 
    $\det(\Sigma)>0$, \cite[Cor. 1]{Stone} states that $\mu_{n}=\Theta(n^{-L/2})$
    provided that $\delta$ is at least as high as the largest span of the
    $a_{l}(X)$ ($X \sim Q$) which are lattice variables.\footnote{\label{foot:Lattice} We say that $X$ is a lattice random variable with offset $\gamma$ and span $h$ if its support is a subset of $\{\gamma+ih\,:\,i\in\ZZ\}$, and the same cannot remain true by increasing $h$.}  
    If all such variables are non-lattice, then $\delta$ can take any positive value.
     
    It only remains to handle the case $\det(\Sigma)=0$.  Suppose that
    $\Sigma$ has rank $L'<L$, and assume without loss of generality that
    $a_1(\cdot),\cdots,a_{L'}(\cdot)$ are linearly independent.  Up to
    sets whose probability with respect to $Q$ is zero, the 
    remaining costs $a_{L'+1}(\cdot),\cdots,a_{L}(\cdot)$ can be written
    as linear combinations of the first $L'$ costs.  Letting $\alpha$ 
    denote the largest magnitude of the scalar coefficients in these linear combinations,
    we conclude that $\xv\in\Dc_n$ provided that 
    \begin{equation} 
        \bigg|\frac{1}{n}\sum_{i=1}^{n}a_{l}(x_{i})-\phi_{l}\bigg|\le\frac{\delta}{\alpha L' n}
    \end{equation}
    for $l=1,\cdots,L'$.  The proposition follows by choosing $\delta$
    to be at least as high as $\alpha L'$ times the largest span of the
    $a_{l}(X)$ which are lattice variables, and analyzing the first $L'$
    costs analogously to the case that $\det(\Sigma)>0$.
\end{IEEEproof}

In accordance with Proposition \ref{prop:SU_SubExpCost}, we henceforth
assume that the choice of $\delta$ for the cost-constrained ensemble
is such that $\mu_{n}=\Omega(n^{-L/2})$.

\section{Random-Coding Error Exponents} \label{sec:SU_ERROR_EXPONENT}

Error exponents characterize the asymptotic exponential behavior of
the error probability in coded communication systems, and can thus
provide additional insight beyond capacity results. In the matched setting,
error exponents were studied by Fano \cite[Ch. 9]{FanoBook} and
later by Gallager \cite[Ch. 5]{Gallager} and Csisz\'ar-K\"orner \cite[Ch. 10]{CsiszarBook}. 
The ensemble tightness of the exponent (cf. \eqref{eq:SU_EnsTight}) under ML decoding
was studied by Gallager \cite{TightAverage} and D'yachkov \cite{DyachkovCC}
for the i.i.d. and constant-composition ensembles respectively.

In this section, we present the ensemble-tight error exponent for
cost-constrained random coding, yielding results for the i.i.d. and 
constant-composition ensembles as special cases. 

\subsection{Cost-Constrained Ensemble}

We define the sets
\begin{align}
    \SetScost(\{a_{l}\}) & \triangleq\big\{ P_{XY}\in\Pc(\Xc\times\Yc)\,:\, \nonumber \\
        & \qquad\quad \EE_{P}[a_{l}(X)]=\phi_{l},\, l=1,\cdots,L\big\}\label{eq:SU_SetS_Cost} \\
    \SetTcost(P_{XY},\{a_{l}\}) & \triangleq\Big\{\Ptilde_{XY}\in\Pc(\Xc\times\Yc)\,:\, \nonumber \\
        & \EE_{\Ptilde}[a_{l}(X)]=\phi_{l}\,\,(l=1,\cdots,L),\Ptilde_{Y}=P_{Y}, \nonumber \\
        & \EE_{\Ptilde}[\log q(X,Y)]\ge\EE_{P}[\log q(X,Y)]\Big\},\label{eq:SU_SetT_Cost}
\end{align}
where the notation $\{a_{l}\}$ is used to denote dependence on $a_{1}(\cdot),\cdots,a_{L}(\cdot)$.
The dependence of these sets on $Q$ (via $\phi_l=\EE_Q[a_l(X)]$) is kept implicit.
\begin{thm} \label{thm:SU_EnsembleExps}
    The random-coding error probability for the cost-constrained ensemble in 
    \eqref{eq:SU_Q_Multi} satisfies
    \begin{equation}
        \lim_{n\to\infty} -\frac{1}{n}\log\pebar(n,e^{nR}) = \Ercost(Q,R,\{a_{l}\}),
    \end{equation}
    where
    \begin{multline}
        \Ercost(Q,R,\{a_{l}\}) \triangleq \min_{P_{XY}\in\SetScost(\{a_{l}\})}\min_{\Ptilde_{XY}\in\SetTcost(P_{XY},\{a_{l}\})} \\
        D(P_{XY}\|Q\times W)+\big[D(\Ptilde_{XY}\|Q\times P_{Y})-R\big]^{+}. \label{eq:SU_Er_Cost} 
    \end{multline}
\end{thm}
\begin{IEEEproof}
    See Appendix \ref{sub:SU_PRIMAL_PROOF}.
\end{IEEEproof}
The optimization problem in \eqref{eq:SU_Er_Cost} 
is convex when the input distribution and auxiliary cost functions are
fixed. The following theorem gives an alternative expression
based on Lagrange duality \cite{Convex}. 
\begin{thm} \label{thm:SU_LagrangeDual}
    The error exponent in \eqref{eq:SU_Er_Cost} can be expressed as
    \begin{equation}
        \Ercost(Q,R,\{a_{l}\}) =\max_{\rho\in[0,1]}\Ezcost(Q,\rho,\{a_{l}\})-\rho R,\label{eq:SU_Er_Cost_Dual} 
    \end{equation}
    where
    \begin{multline}
        \Ezcost(Q,\rho,\{a_{l}\}) \triangleq \sup_{s\ge0,\{r_{l}\},\{\rbar_{l}\}} \\ -\log\EE\left[\bigg(\frac{\EE\big[q(\Xbar,Y)^{s}e^{\sum_{l=1}^{L}\rbar_{l}(a_{l}(\Xbar)-\phi_{l})}\,|\, Y\big]}{q(X,Y)^{s}e^{\sum_{l=1}^{L}r_{l}(a_{l}(X)-\phi_{l})}}\bigg)^{\rho}\right] \label{eq:SU_E0_Cost} 
    \end{multline}
    and $(X,Y,\Xbar)\sim Q(x)W(y|x)Q(\xbar)$.
\end{thm}
\begin{IEEEproof}
    See Appendix \ref{sub:SU_LAGRANGE_PROOF}.
\end{IEEEproof}
The derivation of \eqref{eq:SU_Er_Cost_Dual}--\eqref{eq:SU_E0_Cost} via 
Theorem \ref{thm:SU_EnsembleExps} is useful for proving ensemble tightness, 
but has the disadvantage of being applicable only
in the case of finite alphabets.  We proceed by giving a direct derivation
which does not prove ensemble tightness, but which 
extends immediately to more general alphabets provided that the second
moments associated with the cost functions are finite (see
Proposition \ref{prop:SU_SubExpCost}).  The extension to channels
with input constraints is straightforward; see Section 
\ref{sec:CONCLUSION} for details.

Using Theorem \ref{thm:SU_FiniteUB} and applying $\min\{1,\alpha\}\le\alpha^{\rho}$ ($\rho\in[0,1]$)
to $\rcus$ in \eqref{eq:SU_RCU_s}, we obtain\footnote{In the case of continuous alphabets, the summations should be replaced by integrals.}
\begin{multline}
    \pebar(n,M)\le\frac{1}{\mu_{n}^{1+\rho}}M^{\rho}\sum_{\xv\in\Dc_{n},\yv}Q^{n}(\xv)W^{n}(\yv|\xv) \\ \times\left(\frac{\sum_{\xvbar\in\Dc_{n}}Q^{n}(\xvbar)q^{n}(\xvbar,\yv)^{s}}{q^{n}(\xv,\yv)^{s}}\right)^{\rho},\label{eq:SU_DirectDual3}
\end{multline}
where $Q^{n}(\xv)\triangleq\prod_{i=1}^{n}Q(x_{i})$. 
From \eqref{eq:SU_SetDn}, each codeword $\xv\in\Dc_{n}$ satisfies 
\begin{equation}
    e^{r(a_{l}^{n}(\xv)-n\phi_{l})}e^{|r|\delta}\ge1 \label{eq:SU_DirectDual4}
\end{equation}
for any real number $r$, where $a_{l}^{n}(\xv)\triangleq\sum_{i=1}^{n}a_{l}(x_{i})$.
Weakening \eqref{eq:SU_DirectDual3} by applying \eqref{eq:SU_DirectDual4}
multiple times, we obtain
\begin{multline}
    \pebar(n,M) \le\frac{e^{\rho\sum_{l}(|r_{l}|+|\rbar_{l}|)\delta}}{\mu_{n}^{1+\rho}}M^{\rho}\sum_{\xv\in\Dc_{n},\yv}Q^{n}(\xv)W^{n}(\yv|\xv) \\ \times\left(\frac{\sum_{\xvbar\in\Dc_{n}}Q^{n}(\xvbar)q^{n}(\xvbar,\yv)^{s}e^{\sum_{l}\rbar_{l}(a_{l}^{n}(\xvbar)-n\phi_{l})}}{q^{n}(\xv,\yv)^{s}e^{\sum_{l}r_{l}(a_{l}^{n}(\xv)-n\phi_{l})}}\right)^{\rho}, \label{eq:SU_DirectDual5} 
\end{multline}
where $\{r_{l}\}$ and $\{\rbar_{l}\}$ are arbitrary. Further
weakening \eqref{eq:SU_DirectDual5} by replacing the summations over $\Dc_n$
with summations over all sequences, and expanding 
each term in the outer summation as product from $i=1$ to $n$, we obtain 
\begin{multline}
    \pebar(n,M)\le\frac{e^{\rho\sum_{l}(|r_{l}|+|\rbar_{l}|)\delta}}{\mu_{n}^{1+\rho}}M^{\rho}\Bigg(\sum_{x,y}Q(x)W(y|x) \\ \times\Bigg(\frac{\sum_{\xbar}Q(\xbar)q(\xbar,y)^{s}e^{\sum_{l}\rbar_{l}(a_{l}(\xbar)-\phi_{l})}}{q(x,y)^{s}e^{\sum_{l}r_{l}(a_{l}(x)-\phi_{l})}}\Bigg)^{\rho}\Bigg)^{n}. \label{eq:SU_DirectDual7}
\end{multline}
Since $\mu_{n}$ decays to zero subexponentially in $n$ (cf. Proposition
\ref{prop:SU_SubExpCost}), we conclude that the prefactor in \eqref{eq:SU_DirectDual7}
does not affect the exponent.  Hence, and setting $M=e^{nR}$, we obtain \eqref{eq:SU_Er_Cost_Dual}.

The preceding analysis can be considered a refinement of that of Shamai and Sason \cite{Variations}, 
who showed that an achievable error exponent in the case that $L=1$ is given by
\begin{equation}
    \Ercostprime(Q,R,a_{1})\triangleq \max_{\rho\in[0,1]}\Ezcostprime(Q,\rho,a_{1})-\rho R, \label{eq:SU_Er_Cost'}
\end{equation}
where
\begin{equation}
    \Ezcostprime(Q,\rho,a_{1})\triangleq \sup_{s\ge0}-\log\EE\left[\left(\frac{\EE[q(\Xbar,Y)^{s}e^{a_{1}(\Xbar)}\,|\, Y]}{q(X,Y)^{s}e^{a_{1}(X)}}\right)^{\rho}\right].\label{eq:SU_E0_Cost'} 
\end{equation}
By setting $r_1=\rbar_1=1$ in \eqref{eq:SU_E0_Cost}, we see that $\Ercost$ with $L=1$ 
is at least as  high as $\Ercostprime$.  In Section \ref{sub:SU_EXP_CONNECTIONS}, we show that 
the former can be strictly higher.

\subsection{i.i.d. and Constant-Composition Ensembles} \label{sec:SU_ExpIIDCC}

Setting $L=0$ in \eqref{eq:SU_E0_Cost}, we recover
the exponent of Kaplan and Shamai \cite{Compound}, namely 
\begin{equation}
     \Eriid(Q,R) \triangleq \max_{\rho\in[0,1]}\Eziid(Q,\rho)-\rho R,\label{eq:SU_Er_IID_Dual} 
\end{equation}
where
\begin{equation}
    \Eziid(Q,\rho) \triangleq \sup_{s\ge0}-\log\EE\left[\bigg(\frac{\EE\big[q(\Xbar,Y)^{s}\,|\, Y\big]}{q(X,Y)^{s}}\bigg)^{\rho}\right]. \label{eq:SU_E0_IID}    
\end{equation}
In the special case of constant-composition random coding (see \eqref{eq:SU_Q_CC}--\eqref{eq:SU_CostCaseCC}), 
the constraints $\EE_{\Ptilde}[a_{l}(X)]=\phi_{l}$ for $l=1,\cdots,|\Xc|$
yield $P_X=Q$ and $\Ptilde_X=Q$ in \eqref{eq:SU_SetS_Cost} and \eqref{eq:SU_SetT_Cost}
respectively, and thus \eqref{eq:SU_Er_Cost} recovers Csisz\'{a}r's exponent for constant-composition
coding \cite{Csiszar1}.  Hence, the exponents of 
\cite{Compound,Csiszar1} are tight with respect to the ensemble average.

We henceforth denote the exponent for the constant-composition ensemble
by $\Ercc(Q,R)$.  We claim that 
\begin{equation}
    \Ercc(Q,R) = \max_{\rho\in[0,1]}\Ezcc(Q,\rho)-\rho R\label{eq:SU_Er_LM_Dual},
\end{equation}
where 
\begin{multline}
    \Ezcc(Q,\rho) = \sup_{s\ge0,a(\cdot)} \\ \EE\left[-\log\EE\bigg[\bigg(\frac{\EE\big[q(\Xbar,Y)^{s}e^{a(\Xbar)}\,|\, Y\big]}{q(X,Y)^{s}e^{a(X)}}\bigg)^{\rho}\,\bigg|\, X\bigg]\right]. \label{eq:SU_E0_LM}
\end{multline} 
To prove this, we first note from \eqref{eq:SU_CostCaseCC} that
\begin{align}
    \sum_{l} r_l(a_l(x)-\phi_l) &= \sum_{\xtilde}r_{\xtilde}(\openone\{x=\xtilde\}-Q(\xtilde)) \label{eq:SU_CCDual1} \\
                                &= r(x) - \phi_r, \label{eq:SU_CCDual2}
\end{align}
where \eqref{eq:SU_CCDual1} follows since $\phi_l=\EE_Q[\openone\{x=l\}]=Q(l)$,
and \eqref{eq:SU_CCDual2} follows by defining $r(x)\triangleq r_x$ and $\phi_{r}\triangleq\EE_{Q}[r(X)]$.
Defining $\rbar(x)$ and $\phi_{\rbar}$ similarly, we obtain the following $E_0$ function 
from \eqref{eq:SU_E0_Cost}:
\begin{align}
    \Ezcc(Q,\rho) &\triangleq  \sup_{s\ge0,r(\cdot),\rbar(\cdot)} -\log\sum_{x,y}Q(x)W(y|x) \nonumber \\
                    & \qquad\qquad \times\bigg(\frac{\sum_{\xbar}Q(\xbar)q(\xbar,y)^{s}e^{\rbar(\xbar)-\phi_{\rbar}}}{q(x,y)^{s}e^{r(x)-\phi_r}}\bigg)^{\rho} \label{eq:SU_CCDual3} \\
                  &\le \sup_{s\ge0,r(\cdot),\rbar(\cdot)} -\sum_{x}Q(x)\log\sum_{y}W(y|x) \nonumber \\
                    & \qquad\qquad \times\bigg(\frac{\sum_{\xbar}Q(\xbar)q(\xbar,y)^{s}e^{\rbar(\xbar)-\phi_{\rbar}}}{q(x,y)^{s}e^{r(x)-\phi_r}}\bigg)^{\rho} \label{eq:SU_CCDual4} \\
                  &= \sup_{s\ge0,\rbar(\cdot)} -\sum_{x}Q(x)\log\sum_{y}W(y|x) \nonumber \\
                    & \qquad\qquad \times\bigg(\frac{\sum_{\xbar}Q(\xbar)q(\xbar,y)^{s}e^{\rbar(\xbar)}}{q(x,y)^{s}e^{\rbar(x)}}\bigg)^{\rho}, \label{eq:SU_CCDual5}
\end{align}
where \eqref{eq:SU_CCDual4} follows from Jensen's inequality, and \eqref{eq:SU_CCDual5}
follows by using the definitions of $\phi_r$ and $\phi_{\rbar}$ to write 
\begin{equation}
    -\sum_{x}Q(x)\log\bigg(\frac{e^{-\phi_{\rbar}}}{e^{r(x)-\phi_r}}\bigg)^\rho = -\sum_{x}Q(x)\log\bigg(\frac{1}{e^{\rbar(x)}}\bigg)^\rho.
\end{equation}
Renaming $\rbar(\cdot)$ as $a(\cdot)$, we see that \eqref{eq:SU_CCDual5}
coincides with \eqref{eq:SU_E0_LM}.  It remains to show that equality holds
in \eqref{eq:SU_CCDual4}.  This is easily seen by noting that the choice
\begin{equation}
    r(x) = \frac{1}{\rho}\log\sum_{y}W(y|x)\bigg(\frac{\sum_{\xbar}Q(\xbar)q(\xbar,y)^{s}e^{\rbar(\xbar)}}{q(x,y)^{s}}\bigg)^{\rho}
\end{equation}
makes the logarithm in \eqref{eq:SU_CCDual4} independent 
of $x$, thus ensuring that Jensen's inequality holds with equality.

The exponent $\Eriid(Q,R)$ is positive for all rates below $\GMI(Q)$ \cite{Compound}, whereas $\Ercc$
recovers the stronger rate $\LM(Q)$.  Similarly, both $\Ercost$ ($L=1$) and $\Ercostprime$
recover the LM rate provided that the auxiliary cost is optimized \cite{Variations}.

\subsection{Number of Auxiliary Costs Required} \label{sub:SU_EXP_CONNECTIONS}

We claim that
\begin{equation}
    \Eriid(Q,R)\le \Ercost(Q,R,\{a_{l}\})\le \Ercc(Q,R).\label{eq:SU_Connection1}
\end{equation}
The first inequality follows by setting $r_l=\rbar_l=0$ in \eqref{eq:SU_E0_Cost},
and the second inequality follows by setting $r(x)=\sum_{l}r_l a_l(x)$ and 
$\rbar(x)=\sum_{l}\rbar_l a_l(x)$ in \eqref{eq:SU_E0_Cost}, and upper bounding
the objective by taking the supremum over all $r(\cdot)$ and $\rbar(\cdot)$
to recover $\Ezcc$ in the form given in \eqref{eq:SU_CCDual3}.
Thus, the constant-composition ensemble yields the best error exponent of 
the three ensembles.  

In this subsection, we study the number of auxiliary costs required for
cost-constrained random coding to achieve $\Ercc$. 
Such an investigation is of interest in gaining insight into the codebook structure, 
and since the subexponential prefactor in \eqref{eq:SU_DirectDual7} 
grows at a slower rate when $L$ is reduced (see Proposition \ref{prop:SU_SubExpCost}).
Our results are summarized in the following theorem.

\begin{thm} \label{thm:SU_NumCosts}
    Consider a DMC $W$ and input distribution $Q$.
    \begin{enumerate} 
      \item For any decoding metric, we have
            \begin{align}
                & \sup_{a_{1}(\cdot),a_{2}(\cdot)}\Ercost(Q,R,\{a_{1},a_{2}\}) = \Ercc(Q,R) \label{eq:SU_Connection2} \\
                & \max_Q\sup_{a_1(\cdot)}\Ercostprime(Q,R,a_1) = \max_Q\Ercc(Q,R). \label{eq:SU_MaxQ}
            \end{align} 
      \item If $q(x,y)=W(y|x)$ (ML decoding), then
            \begin{align}
                & \,\sup_{a_{1}(\cdot)}\Ercost(Q,R,a_{1}) = \Ercc(Q,R) \label{eq:SU_Connection2_ML} \\
                & \sup_{a_1(\cdot)}\Ercostprime(Q,R,a_1)=\Eriid(Q,R) \label{eq:SU_ConnML2} \\
                & \max_Q\Eriid(Q,R) = \max_Q\Ercc(Q,R). \label{eq:SU_MaxQ_ML}
            \end{align} 
    \end{enumerate}
\end{thm}
\begin{proof}
    We have from \eqref{eq:SU_Connection1} that $\Ercost \le \Ercc$.  To obtain
    the reverse inequality corresponding to \eqref{eq:SU_Connection2}, 
    we set $L=2$, $r_1=\rbar_2=1$ and $r_2=\rbar_1=0$ in 
    \eqref{eq:SU_E0_Cost}.  The resulting objective coincides with 
    \eqref{eq:SU_CCDual3} upon setting $a_1(\cdot)=r(\cdot)$ and 
    $a_2(\cdot)=\rbar(\cdot)$.
    
    To prove \eqref{eq:SU_MaxQ}, we note the following observation 
    from Appendix \ref{sub:SU_NECC_CONDS}: Given $s>0$ and $\rho>0$, 
    any pair $(Q,a)$ maximizing the objective in \eqref{eq:SU_E0_LM} must
    satisfy the property that the logarithm
    in \eqref{eq:SU_E0_LM} has the same value for all $x$ such that $Q(x)>0$. 
    It follows that the objective in \eqref{eq:SU_E0_LM} is unchanged when
    the expectation with respect to $X$ is moved inside the logarithm, thus 
    yielding the objective in \eqref{eq:SU_E0_Cost'}.

    We now turn to the proofs of \eqref{eq:SU_Connection2_ML}--\eqref{eq:SU_MaxQ_ML}.
    We claim that, under ML decoding, we can write $\Ezcc$ as
    \begin{multline}
        \Ezcc(Q,\rho) = \sup_{a(\cdot)} \\ -\log\sum_{y}\bigg(\sum_{x}Q(x)W(y|x)^{\frac{1}{1+\rho}}e^{a(x)-\phi_a}\bigg)^{1+\rho}, \label{eq:SU_ConnML1}
    \end{multline}
    where $\phi_a \triangleq \EE_Q[a(X)]$.  To show this,
    we make use of the form of $\Ezcc$ given in \eqref{eq:SU_CCDual3},
    and write the summation inside the logarithm as
    \begin{multline}
        \sum_{y}\bigg(\sum_{x}Q(x)W(y|x)^{1-s\rho}e^{-\rho(r(x)-\phi_r)}\bigg) \\ \times \bigg(\sum_{\xbar}Q(\xbar)W(y|\xbar)^{s}e^{\rbar(\xbar)-\phi_{\rbar}}\bigg)^\rho.
    \end{multline}
    Using H\"{o}lder's inequality in an identical fashion to \cite[Ex. 5.6]{Gallager},
    this summation is lower bounded by
    \begin{equation}
        \sum_{y}\bigg(\sum_{x}Q(x)W(y|x)^{\frac{1}{1+\rho}}e^{\rbar(x)-\phi_{\rbar}}\bigg)^{1+\rho}
    \end{equation}
    with equality if and only if $s=\frac{1}{1+\rho}$ and $\rbar(\cdot)= -\rho r(\cdot)$.  Renaming
    $\rbar(\cdot)$ as $a(\cdot)$, we obtain \eqref{eq:SU_ConnML1}.  We can clearly achieve $\Ercc$ 
    using $L=2$ with the cost functions $r(\cdot)$ and $\rbar(\cdot)$.  However, since we have
    shown that one is a scalar multiple of the other, we conclude that $L=1$ suffices. 
    
    A similar argument using H\"{o}lder's inequality reveals that the objective in \eqref{eq:SU_E0_Cost'}
    is maximized by $s=\frac{1}{1+\rho}$ and $a_1(\cdot)=0$, and the objective in \eqref{eq:SU_E0_IID}
    is maximized by $s=\frac{1}{1+\rho}$, thus yielding \eqref{eq:SU_ConnML2}.
    Finally, combining \eqref{eq:SU_MaxQ} and \eqref{eq:SU_ConnML2}, we obtain 
    \eqref{eq:SU_MaxQ_ML}.
\end{proof}

Theorem \ref{thm:SU_NumCosts} shows that the cost-constrained ensemble recovers $\Ercc$ using at
most two auxiliary costs.  If either the input distribution or decoding rule
is optimized, then $L=1$ suffices (see \eqref{eq:SU_MaxQ} and \eqref{eq:SU_Connection2_ML}), and if both are optimized 
then $L=0$ suffices (see \eqref{eq:SU_MaxQ_ML}).  The latter result is well-known
\cite{GallagerCC} and is stated for completeness.
While \eqref{eq:SU_MaxQ} shows that $\Ercost$ and $\Ercostprime$ coincide when
$Q$ is optimized, \eqref{eq:SU_Connection2_ML}--\eqref{eq:SU_ConnML2} show 
that the former can be strictly higher for a given $Q$ even when $L=1$,
since $\Ercc$ can exceed $\Eriid$ even under ML decoding \cite{GallagerCC}. 

\subsection{Numerical Example} \label{sub:SU_NUMERICAL}

We consider the channel defined by the entries of the $|\Xc|\times|\Yc|$ matrix 
\begin{equation}
    \left[\begin{array}{ccc}
        1-2\delta_{0} & \delta_{0} & \delta_{0}\\
        \delta_{1} & 1-2\delta_{1} & \delta_{1}\\
        \delta_{2} & \delta_{2} & 1-2\delta_{2}
    \end{array}\right]\label{eq:SU_MatrixW}
\end{equation}
with $\Xc=\Yc=\{0,1,2\}$. The mismatched decoder
chooses the codeword which is closest to $\yv$ in terms
of Hamming distance. For example, the decoding metric can be taken
to be the entries of \eqref{eq:SU_MatrixW} with
$\delta_i$ replaced by $\delta\in(0,\frac{1}{3})$ for $i=1,2,3$.
We let $\delta_{0}=0.01$, $\delta_{1}=0.05$, $\delta_{2}=0.25$
and $Q=(0.1,0.3,0.6)$. Under these parameters, we have 
$\GMI(Q)=0.387$, $\LM(Q)=0.449$ and $I(X;Y)=0.471$ bits/use. 

We evaluate the exponents using the optimization software YALMIP \cite{YALMIP}.
For the cost-constrained ensemble with $L=1$, we optimize the auxiliary cost.
As expected, Figure \ref{fig:SU_Exponents} shows that 
the highest exponent is $\Ercc$. The exponent 
$\Ercost$ ($L=1$) is only marginally lower than $\Ercc$, whereas the gap to 
$\Ercostprime$ is larger. The exponent $\Eriid$ is not only lower
than each of the other exponents, but also yields a worse achievable rate.
In the case of ML decoding, $\Ercc$ exceeds $\Eriid$ for all $R<I(X;Y)$.

\begin{figure}
    \begin{centering}
        \includegraphics[width=1\columnwidth]{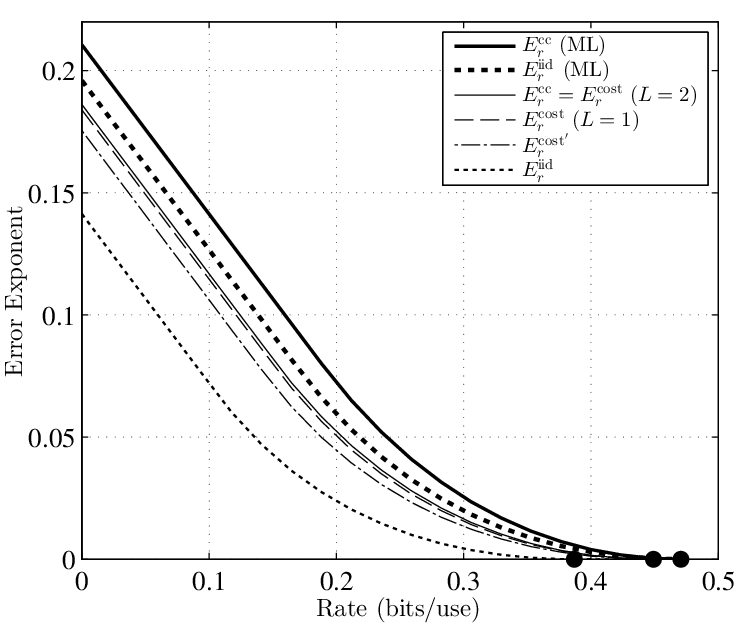}
        \par
    \end{centering}
 
    \caption{Error exponents for the channel defined in \eqref{eq:SU_MatrixW}
        with $\delta_{0}=0.01$, $\delta_{1}=0.05$, $\delta_{2}=0.25$ and
        $Q=(0.1,0.3,0.6)$. The mismatched decoder uses the minimum Hamming
        distance metric. The corresponding achievable rates $\GMI(Q)$, $\LM(Q)$ and
        $I(X;Y)$ are respectively marked on the horizontal axis.\label{fig:SU_Exponents}}
\end{figure}

\section{Second-Order Coding Rates} \label{sec:SU_SECOND_ORDER_RATES}

In the matched setting, the finite-length performance limits of a
channel are characterized by $M^{*}(n,\epsilon)$, defined to be the
maximum number of codewords of length $n$ yielding an error probability
not exceeding $\epsilon$ for some encoder and decoder. The
problem of finding the second-order asymptotics
of $M^{*}(n,\epsilon)$ for a given $\epsilon$ was studied
by Strassen \cite{Strassen}, and later revisited by Polyanskiy \emph{et al.} 
\cite{Finite} and Hayashi \cite{Hayashi}, among others. For
DMCs, we have under mild technical conditions that 
\begin{equation}
    \log M^{*}(n,\epsilon)=nC-\sqrt{nV}\,\Qsf^{-1}(\epsilon)+O(\log n),\label{eq:SU_SecondOrder}
\end{equation}
where $C$ is the channel capacity, and $V$ is known as the channel dispersion.
Results of the form \eqref{eq:SU_SecondOrder} provide a quantification of
the speed of convergence to the channel capacity as the block length increases.

In this section, we present achievable second-order coding rates for
the ensembles given in Section \ref{sec:INTRO}, i.e. expansions of
the form \eqref{eq:SU_SecondOrder} with the equality replaced by
$\ge$. To distinguish between the ensembles, we define $\Miid(Q,n,\epsilon)$,
$\Mcc(Q,n,\epsilon)$ and $\Mcost(Q,n,\epsilon)$
to be the maximum number of codewords of length $n$ such that the
random-coding error probability does not exceed $\epsilon$ for the
i.i.d., constant-composition and cost-constrained
ensembles respectively, using the input distribution $Q$. 
We first consider the discrete memoryless setting, and then discuss
more general memoryless channels.

\subsection{Cost-Constrained Ensemble} \label{sub:SU_CostCostr2OR}

A key quantity in the second-order analysis for ML decoding is the
information density, given by
\begin{equation}
    i(x,y)\triangleq\log\frac{W(y|x)}{\sum_{x}Q(x)W(y|x)},
\end{equation}
where $Q$ is a given input distribution. In the mismatched setting,
the relevant generalization of $i(x,y)$ is
\begin{equation}
    i_{s,a}(x,y) \triangleq\log\frac{q(x,y)^{s}e^{a(x)}}{\sum_{\xbar}Q(\xbar)q(\xbar,y)^{s}e^{a(\xbar)}},\label{eq:SU_i_sa} 
\end{equation}
where $s\ge0$ and $a(\cdot)$ are fixed parameters. We write
$i_{s,a}^{n}(\xv,\yv) \triangleq\sum_{i=1}^{n}i_{s,a}(x_{i},y_{i})$
and similarly $Q^{n}(\xv)\triangleq\prod_{i=1}^{n}Q(x_{i})$
and $a^{n}(\xv)\triangleq\sum_{i=1}^{n}a(x_{i})$. We define
\begin{align}
    I_{s,a}(Q) & \triangleq \EE[i_{s,a}(X,Y)] \label{eq:SU_Isa} \\
    U_{s,a}(Q) & \triangleq \var[i_{s,a}(X,Y)] \label{eq:SU_Usa} \\
    V_{s,a}(Q) & \triangleq \EE\big[\var[i_{s,a}(X,Y)\,|\, X]\big], \label{eq:SU_Vsa}
\end{align}
where $(X,Y)\sim Q\times W$. From \eqref{eq:INTR_RateLM},
we see that the LM rate is equal to $I_{s,a}(Q)$ after optimizing $s$ and $a(\cdot)$.

We can relate \eqref{eq:SU_Isa}--\eqref{eq:SU_Vsa} with the 
$E_0$ functions defined in \eqref{eq:SU_E0_Cost'} and \eqref{eq:SU_E0_LM}.
Letting $\Ezcostprime(Q,\rho,s,a)$ and $\Ezcc(Q,\rho,s,a)$ denote
the corresponding objectives with fixed $(s,a)$ in place of the
supremum, we have $I_{s,a}=\frac{\partial\Ezcostprime}{\partial\rho}\Big|_{\rho=0}=\frac{\partial\Ezcc}{\partial\rho}\Big|_{\rho=0}$,
$U_{s,a}=-\frac{\partial^2\Ezcostprime}{\partial\rho^2}\Big|_{\rho=0}$,
and $V_{s,a}=-\frac{\partial^2\Ezcc}{\partial\rho^2}\Big|_{\rho=0}$.
The latter two identities generalize a well-known connection between the 
exponent and dispersion in the matched case \cite[p. 2337]{Finite}.

The main result of this subsection is the following theorem, which considers
the cost-constrained ensemble.  Our proof differs from the usual
proof using threshold-based random-coding bounds \cite{Strassen,Finite},
but the latter approach can also be used in the present setting \cite{PaperITA}.
Our analysis can be interpreted as performing a normal approximation
of $\rcus$ in \eqref{eq:SU_RCU_s}.
\begin{thm} \label{thm:SU_CostDispersion}
    Fix the input distribution $Q$ and the parameters $s\ge0$ and $a(\cdot)$.
    Using the cost-constrained ensemble in \eqref{eq:SU_Q_Multi} with $L=2$ and
    \begin{align}
        a_{1}(x) & =a(x)\label{eq:SU_a1} \\
        a_{2}(x) & = \EE_{W(\cdot|x)}[i_{s,a}(x,Y)] \label{eq:SU_a2},
    \end{align}
    the following expansion holds: 
    \begin{equation}
        \log \Mcost(Q,n,\epsilon)\ge nI_{s,a}(Q)-\sqrt{nV_{s,a}(Q)}\Qsf^{-1}(\epsilon)+O(\log n). \label{eq:SU_DispersionCost}
    \end{equation}
\end{thm} 
\begin{proof}
    Throughout the proof, we make use of the random variables $(X,Y,\Xbar)\sim Q(x)W(y|x)Q(\xbar)$
    and $(\Xv,\Yv,\Xvbar)\sim P_{\Xv}(\xv)W^{n}(\yv|\xv)P_{\Xv}(\xvbar)$.
    Probabilities, expectations, etc. containing a realization $\xv$ of $\Xv$
    are implicitly defined to be conditioned on the event $\Xv=\xv$.

    We start with Theorem \ref{thm:SU_FiniteUB} and weaken $\rcus$ in \eqref{eq:SU_RCU_s} as follows:
    \begin{align}
        & \rcus(n,M) \nonumber
            \\ & =\EE\bigg[\min\bigg\{1,(M-1)\frac{\sum_{\xvbar \in \Dc_n}P_{\Xv}(\xvbar)q^{n}(\Xvbar,\Yv)^{s}}{q^{n}(\Xv,\Yv)^{s}}\bigg\}\bigg] \\
            & \le\EE\bigg[\min\bigg\{1,Me^{2\delta}\frac{\sum_{\xvbar\in \Dc_n}P_{\Xv}(\xvbar)q^{n}(\Xvbar,\Yv)^{s}e^{a^{n}(\Xvbar)}}{q^{n}(\Xv,\Yv)^{s}e^{a^{n}(\Xv)}}\bigg\}\bigg]\label{eq:SU_Disp1_2} \\
            & \le\EE\bigg[\min\bigg\{1,\frac{Me^{2\delta}}{\mu_{n}}\frac{\sum_{\xvbar}Q^{n}(\xvbar)q^{n}(\Xvbar,\Yv)^{s}e^{a^{n}(\Xvbar)}}{q^{n}(\Xv,\Yv)^{s}e^{a^{n}(\Xv)}}\bigg\}\bigg]\label{eq:SU_Disp1_3} \\
            & =\PP\bigg[i_{s,a}^{n}(\Xv,\Yv)+\log U\le\log\frac{Me^{2\delta}}{\mu_{n}}\bigg] \label{eq:SU_Disp1_4} \\
            & \le \PP\bigg[i_{s,a}^{n}(\Xv,\Yv)+\log U\le\log\frac{Me^{2\delta}}{\mu_{n}} \,\cap\, \Xv\in\Ac_n \bigg] \nonumber \\
                & \qquad\qquad\qquad\qquad\qquad\qquad\qquad\qquad + \PP\big[\Xv \notin \Ac_n\big] \label{eq:SU_Disp1_4a} \\
            & \le \max_{\xv\in\Ac_{n}} \PP\bigg[i_{s,a}^{n}(\xv,\Yv)+\log U\le\log\frac{Me^{2\delta}}{\mu_{n}}\bigg] + \PP\big[\Xv \notin \Ac_n\big],\label{eq:SU_Disp1_5} 
    \end{align}
    where \eqref{eq:SU_Disp1_2} follows from \eqref{eq:SU_DirectDual4}, \eqref{eq:SU_Disp1_3} 
    follows by substituting the random-coding distribution in \eqref{eq:SU_Q_Multi} and summing
    over all $\xvbar$ instead of $\xvbar\in\Dc_{n}$, and \eqref{eq:SU_Disp1_4} follows from the 
    definition of $i^n_{s,a}$ and the identity
    \begin{equation}
        \EE[\min\{1,A\}]=\PP[A>U], \label{eq:SU_UnifIneq}
    \end{equation}
    where $A$ is an arbitrary non-negative random variable, and $U$ is 
    uniform on $(0,1)$ and independent of $A$.  Finally, \eqref{eq:SU_Disp1_4a}
    holds for any set $\Ac_n$ by the law of total probability.
 
    We treat the cases $V_{s,a}(Q)>0$ and $V_{s,a}(Q)=0$ separately.
    In the former case, we choose
    \begin{equation}
        \Ac_n = \bigg\{ \xv \in \Dc_n \,:\, \Big|\frac{1}{n}v_{s,a}^n(\xv) - V_{s,a}(Q)\Big| \le \zeta\sqrt{\frac{\log n}{n}} \bigg\}, \label{eq:SU_SetAn}
    \end{equation}
    where $\zeta$ is a constant, and $v_{s,a}^n(\xv)\triangleq\sum_{i=1}^{n}v_{s,a}(x_i)$ with
    \begin{equation}
        v_{s,a}(x) \triangleq \var_{W(\cdot|x)}[i_{s,a}(x,Y)]. \label{eq:SU_vsa}
    \end{equation} 
    Using this definition along with that of $\Dc_{n}$ in \eqref{eq:SU_SetDn} and  
    the cost function in \eqref{eq:SU_a2}, we have for any $\xv\in\Ac_n$ that
    \begin{align}
        \Big|\EE[i_{s,a}^{n}(\xv,\Yv)] - nI_{s,a}(Q)\Big| &\le \delta \label{eq:SU_Mean} \\
        \Big|\var[i_{s,a}^{n}(\xv,\Yv)] - nV_{s,a}(Q)\Big| &\le \zeta \sqrt{n\log n}, \label{eq:SU_Var} 
    \end{align}
    for all $\xv\in\Dc_n$. Since $\log U$ has finite moments, this implies
    \begin{align}
        \Big|\EE[i_{s,a}^{n}(\xv,\Yv)+\log U] - nI_{s,a}(Q)\Big| &= O(1) \label{eq:SU_Mean-1} \\
        \Big|\var[i_{s,a}^{n}(\xv,\Yv)+\log U] - nV_{s,a}(Q)\Big| &= O\big(\sqrt{n\log n}\big). \label{eq:SU_Var-1} 
    \end{align}
    
    Using \eqref{eq:SU_Q_Multi} and defining $\Xv' \sim Q^n(\xv')$, we have
    \begin{align}
        \PP\big[\Xv \notin \Ac_n\big] \le \frac{1}{\mu_n} \PP\big[\Xv' \notin \Ac_n\big]. \label{eq:SU_SetAc1}
    \end{align}
    We claim that there exists a choice of $\zeta$ such that the right-hand side 
    of \eqref{eq:SU_SetAc1} behaves as $O\big(\frac{1}{\sqrt n}\big)$, thus yielding
    \begin{equation}
        \PP\big[\Xv \notin \Ac_n\big] = O\Big(\frac{1}{\sqrt n}\Big). \label{eq:SU_SetAc2}
    \end{equation} 
    Since Proposition \ref{prop:SU_SubExpCost}
    states that $\mu_n=\Omega(n^{-L/2})$, it suffices to show that $\PP[\Xv' \notin \Ac_n]$ 
    can be made to behave as $O(n^{-(L+1)/2})$.  This 
    follows from the following moderate deviations result of \cite[Thm. 2]{ModDevMaths}:
    Given an i.i.d. sequence $\{Z_i\}_{i=1}^{n}$ with $\EE[Z_i]=\mu$ and $\var[Z_i]=\sigma^2 > 0$,
    we have $\PP\big[\big|\frac{1}{n}\sum_{i=1}^{n}Z_i - \mu\big| > \eta\sigma\sqrt{\frac{\log n}{n}}\big] \asymp \frac{2}{\eta\sqrt{2\pi\log n}}n^{-\eta^2/2}$
    provided that $\EE[Z_i^{\eta^2+2+\delta}]<\infty$ for some $\delta>0$.  The latter condition
    is always satisfied in the present setting, since we are considering finite alphabets.

    We are now in a position to apply the Berry-Esseen theorem for independent and 
    non-identically distributed random variables \cite[Sec. XVI.5]{Feller}. The relevant first
    and second moments are bounded in \eqref{eq:SU_Mean-1}--\eqref{eq:SU_Var-1},
    and the relevant third moment is bounded since we are considering finite alphabets.
    Choosing
    \begin{equation}
        \log M=nI_{s,a}(Q)-\log\mu_{n}-2\delta-\xi_n \label{eq:SU_Disp_0_2}
    \end{equation} 
    for some $\xi_n$, and also using \eqref{eq:SU_Disp1_5} and \eqref{eq:SU_SetAc2},
    we obtain from the Berry-Esseen theorem that
    \begin{equation}
        \pebar\le\Qsf\Bigg(\frac{\xi_n+O(1)}{\sqrt{nV_{s,a}(Q)+O(\sqrt{n \log n})}}\Bigg)+O\Big(\frac{1}{\sqrt{n}}\Big).\label{eq:SU_Disp_0_4}
    \end{equation}
    By straightforward rearrangements and a first-order Taylor expansion of 
    the square root function and the $\Qsf^{-1}$ function, we obtain
    \begin{equation}
        \xi_n \le\sqrt{nV_{s,a}(Q)}\,\Qsf^{-1}(\pebar)+O\big(\sqrt{\log n}\big). \label{eq:SU_Disp_0_5}
    \end{equation}
    The proof for the case $V_{s,a}(Q)>0$  is concluded by combining \eqref{eq:SU_Disp_0_2} and 
    \eqref{eq:SU_Disp_0_5}, and noting from Proposition \ref{prop:SU_SubExpCost} that $\log\mu_{n}=O(\log n)$.
    
    In the case that $V_{s,a}(Q)=0$, we can still make use of \eqref{eq:SU_Mean-1}, but
    the variance is handled differently.  From the definition in \eqref{eq:SU_Vsa},
    we in fact have $\var[i_{s,a}(x,Y)]=0$ for all $x$ such that $Q(x)>0$.  Thus,
    for all $\xv\in\Dc_n$ we have $\var[i_{s,a}^{n}(\xv,\Yv)] = 0$ and hence
    $\var[i_{s,a}^{n}(\xv,\Yv) + \log U] = O(1)$.  Choosing $M$ as in 
    \eqref{eq:SU_Disp_0_2} and setting $\Ac_n=\Dc_n$, we can write \eqref{eq:SU_Disp1_5} as
    \begin{align}
        \rcus(n,M) &\le \max_{\xv\in\Dc_{n}} \PP\big[i_{s,a}^{n}(\xv,\Yv)+\log U - nI_{s,a}(Q) \le -\xi_n\big] \\
                   &\le \frac{O(1)}{(\xi_n - O(1))^2}, \label{eq:SU_Disp_0_6}
    \end{align}
    where \eqref{eq:SU_Disp_0_6} holds due to \eqref{eq:SU_Mean-1} and Chebyshev's inequality provided that 
    $\xi_n$ is sufficiently large so that the $\xi_n-O(1)$ term is positive.
    Rearranging, we see that we can achieve any target value $\pebar=\epsilon$ with
    $\xi_n = O(1)$.  The proof is concluded using \eqref{eq:SU_Disp_0_2}.
\end{proof}

Theorem \ref{thm:SU_CostDispersion} can easily be extended to channels
with more general alphabets.   However, some care is needed, since
the moderate deviations result \cite[Thm. 2]{ModDevMaths} used in the proof requires
finite moments up to a certain order depending on $\zeta$ in \eqref{eq:SU_SetAn}. 
In the case that \emph{all} moments of $i_{s,a}(X,Y)$ are finite, the preceding analysis 
is nearly unchanged, except that the third moment 
should be bounded in the set $\Ac_n$ in \eqref{eq:SU_SetAn} in the same way
as the second moment.  An alternative approach is to introduce two further
auxiliary costs into the ensemble:
\begin{align}
    a_3(x) &= v_{s,a}(x) \\
    a_4(x) &= \EE\big[|i_{s,a}(x,Y) - I_{s,a}(Q)|^3\big],
\end{align}
where $v_{s,a}$ is defined in \eqref{eq:SU_vsa}.  Under these choices, 
the relevant second and third moments 
for the Berry-Esseen theorem are bounded within $\Dc_n$ similarly to \eqref{eq:SU_Mean-1}.
The only further requirement is that the sixth moment of $i_{s,a}(X,Y)$ is finite
under $Q \times W$, in accordance with Proposition \ref{prop:SU_SubExpCost}.

We can easily deal with additive input constraints
by handling them similarly to the auxiliary costs (see Section \ref{sec:CONCLUSION}
for details).  With these modifications, our techniques
provide, to our knowledge, the most general known second-order achievability proof
for memoryless input-constrained channels with infinite or continuous alphabets.\footnote{Analogous results were stated in \cite{Hayashi}, but the generality of the proof techniques therein is unclear.  In particular, the quantization arguments on page 4963 therein require that the rate of convergence from $I(X_m;Y)$ to $I(X;Y)$ is sufficiently fast, where $X_m$ is the quantized input variable with a support of cardinality $m$.}
In particular, for the additive white Gaussian noise (AWGN) 
channel with a maximal power constraint and ML decoding, setting $s=1$ and $a(\cdot)=0$ yields the 
achievability part of the dispersion given by Polyanskiy \emph{et al.} \cite{Finite},
thus providing a simple alternative to the proof therein based on the $\kappa\beta$ bound. 

\subsection{i.i.d. and Constant-Composition Ensembles} \label{eq:SU_2ORCCIID}

The properties of the cost-constrained ensemble used in the proof of Theorem 
\ref{thm:SU_CostDispersion} are also satisfied by the constant-composition
ensemble, so we conclude that \eqref{eq:SU_DispersionCost} remains true
when $\Mcost$ is replaced by $\Mcc$.  However, using standard bounds on $\mu_n$
in \eqref{eq:SU_Disp1_5} (e.g. \cite[p. 17]{CsiszarBook}), we obtain a third-order
$O(\log n)$ term which grows linearly in $|\Xc|$.  In contrast,
by Proposition \ref{prop:SU_SubExpCost} and \eqref{eq:SU_Disp_0_2}, the cost-constrained 
ensemble yields a third-order term of the form $-\frac{L}{2}\log n + O(1)$,
where $L$ is independent of $|\Xc|$.

The second-order asymptotic result for i.i.d. coding does not follow directly
from  Theorem \ref{thm:SU_CostDispersion}, since the proof requires the
cost function in \eqref{eq:SU_a2} to be present.  However,
using similar arguments along with the identities $\EE[i_s^n(\Xv,\Yv)]=nI_s(Q)$  
and $\var[i_s^n(\Xv,\Yv)]=nU_s(Q)$ (where $\Xv \sim Q^n$), we obtain
\begin{equation} 
    \log \Miid(Q,n,\epsilon)\ge nI_{s}(Q)-\sqrt{nU_{s}(Q)}\,\Qsf^{-1}(\epsilon)+O(1) \label{eq:SU_DispersionIID}
\end{equation} 
for $s\ge0$, where $I_s(Q)$ and $U_s(Q)$ are defined as
in \eqref{eq:SU_Isa}--\eqref{eq:SU_Usa} with $a(\cdot)=0$.  Under some technical
conditions, the $O(1)$ term in \eqref{eq:SU_DispersionIID} can be improved
to $\frac{1}{2}{\log n}+O(1)$ using the techniques of \cite[Sec. 3.4.5]{FiniteThesis};
see Section \ref{sub:SU_PREFACTOR_RCU} for further discussion.

\subsection{Number of Auxiliary Costs Required} \label{sec:SU_2OR_NCOSTS}

For ML decoding ($q(x,y)=W(y|x)$), we immediately see that $a_1(\cdot)$ in \eqref{eq:SU_a1} 
is not needed, since the parameters maximizing $I_{s,a}(Q)$ in \eqref{eq:SU_Isa}
are $s=1$ and $a(\cdot)=0$, thus yielding the mutual information.

We claim that, for any decoding metric, the auxiliary cost $a_2(\cdot)$ in  
\eqref{eq:SU_a2} is not needed in the case that $Q$ and $a(\cdot)$ are 
optimized in \eqref{eq:SU_DispersionCost}.  This follows from the following
observation proved in Appendix \ref{sub:SU_NECC_CONDS}: 
Given $s>0$, any pair $(Q,a)$ which maximizes $I_{s,a}(Q)$ must 
be such that $\EE_{W(\cdot|x)}[i_{s,a}(x,Y)]$ has  
the same value for all $x$ such that $Q(x)>0$.  Stated differently, the conditional variance
$V_{s,a}(Q)$ coincides with the unconditional variance $U_{s,a}(Q)$ after
the optimization of the parameters, thus generalizing the analogous
result for ML decoding \cite{Finite}.
    
We observe that the number of auxiliary costs in each case coincides
with that of the random-coding exponent (see Section \ref{sub:SU_EXP_CONNECTIONS}):
$L=2$ suffices in general, $L=1$ suffices if the metric or input distribution
is optimized, and $L=0$ suffices is both are optimized.

\section{Saddlepoint Approximations} \label{sec:SU_REFINED_IID}

Random-coding error exponents can be thought of as providing an estimate
of the error probability of the form $\pe \approx e^{-nE_r(R)}$.  More refined
estimates can be obtained having the form $\pe \approx \alpha_n(R)e^{-nE_r(R)}$,
where $\alpha_n(R)$ is a subexponential prefactor.
Early works on characterizing the subexponential prefactor for a given rate under 
ML decoding include those of Elias \cite{TwoChannels} and Dobrushin
\cite{Dobrushin}, who studied specific channels exhibiting
a high degree of symmetry.  More recently, Altu\u{g} and Wagner \cite{RefinementJournal,RefinementSP}
obtained asymptotic prefactors for arbitrary DMCs. 

In this section, we take an alternative approach based on the  
saddlepoint approximation \cite{SaddlepointBook}.  Our goal is to 
provide approximations for $\rcu$ and $\rcus$ (see Theorem \ref{thm:SU_FiniteUB}) which are not only tight 
in the limit of large $n$ for a fixed rate, but also when the 
rate varies.  In particular, our analysis will cover the regime 
of a fixed target error probability, which was studied in 
Section \ref{sec:SU_SECOND_ORDER_RATES}, as well as the moderate deviations
regime, which was studied in \cite{ModerateDev,ModerateDev2}. We focus on i.i.d. 
random coding, which is particularly amenable to a precise asymptotic analysis.

\subsection{Preliminary Definitions and Results} \label{sec:SU_REF_DEFS} 

Analogously to Section \ref{sec:SU_SECOND_ORDER_RATES}, we fix $Q$ 
and $s>0$ and define the quantities
\begin{align}
    i_s(x,y)       &\triangleq \log\frac{q(x,y)^s}{\sum_{\xbar}Q(\xbar)q(\xbar,y)^s} \label{eq:SU_is} \\
    i_s^n(\xv,\yv) &\triangleq \sum_{i=1}^{n}i_s(x_i,y_i) \label{eq:SU_is_n} \\
    I_s(Q)         &\triangleq \EE[i_s(X,Y)] \label{eq:SU_Is} \\
    U_s(Q)         &\triangleq \var[i_s(X,Y)] \label{eq:SU_Us},
\end{align}
where $(X,Y)\sim Q\times W$. We write $\rcus$ in \eqref{eq:SU_RCU_s} (with $P_{\Xv}=Q^n$) as
\begin{equation}
    \rcus(n,M) = \EE\Big[\min\big\{1,(M-1)e^{-i_s^n(\Xv,\Yv)}\big\}\Big]. \label{eq:SU_RCU_s2}
\end{equation}
We let 
\begin{equation}
    \Eziid(Q,\rho,s) \triangleq -\log\EE\big[e^{-\rho i_s(X,Y)}\big] \label{eq:SU_E0s_IID}
\end{equation} 
denote the objective in \eqref{eq:SU_E0_IID} with a fixed value of
$s$ in place of the supremum.  The optimal value of $\rho$ is given by
\begin{equation}
    \rhohat(Q,R,s) \triangleq \argmax_{\rho\in[0,1]} \Eziid(Q,\rho,s) - \rho R. \label{eq:SU_rho_hat}
\end{equation} 
and the critical rate is defined as
\begin{equation}
    \Rcrs(Q) \triangleq \sup\big\{R\,:\,\rhohat(Q,R,s)=1\big\}.
\end{equation}
Furthermore, we define the following derivatives associated with \eqref{eq:SU_rho_hat}:
\begin{align}
    c_{1}(Q,R,s) & \triangleq R-\frac{\partial \Eziid(Q,\rho,s)}{\partial\rho}\bigg|_{\rho=\rhohat(Q,R,s)}\label{eq:SA_c1} \\
    c_{2}(Q,R,s) & \triangleq-\frac{\partial^{2}\Eziid(Q,\rho,s)}{\partial\rho^{2}}\bigg|_{\rho=\rhohat(Q,R,s)}, \label{eq:SA_c2} 
\end{align}

The following properties of the above quantities are analogous to those of Gallager
for ML decoding \cite[pp. 141-143]{Gallager}, and can be proved in a similar fashion: 
\begin{enumerate}
  \item For all $R\ge0$, we have $c_{2}(Q,R,s) > 0$ if $U_s(Q)>0$, and $c_{2}(Q,R,s) = 0$ if $U_s(Q)=0$.  Furthermore, we have $c_2(Q,I_s(Q),s)=U_s(Q)$.
  \item If $U_s(Q)=0$, then $\Rcrs(Q)=I_s(Q)$.
  \item For $R \in \big[0,\Rcrs(Q)\big)$, we have $\rhohat(Q,R,s)=1$ and $c_{1}(Q,R,s)<0$.
  \item For $R \in \big[\Rcrs(Q),I_s(Q)\big]$, $\rhohat(Q,R,s)$ is strictly decreasing in $R$,
        and $c_{1}(Q,R,s)=0$.
  \item For $R > I_s(Q)$, we have $\rhohat(Q,R,s)=0$ and $c_{1}(Q,R,s)>0$.
\end{enumerate}
Throughout this section, the arguments to $\rhohat$, $c_1$, etc. will be
omitted, since their values will be clear from the context. 

The density function of a $N(\mu,\sigma^2)$ random variable is denoted by
\begin{equation}
    \phi(z;\mu,\sigma^2) \triangleq \frac{1}{\sqrt{2\pi\sigma^2}}e^{-\frac{(z-\mu)^2}{2\sigma^2}}. \label{eq:SU_phi}
\end{equation} 
When studying lattice random variables (see Footnote \ref{foot:Lattice} on Page \pageref{foot:Lattice}) 
with span $h$, it will be useful to define
\begin{equation}
    \phi_h(z;\mu,\sigma^2) \triangleq \frac{h}{\sqrt{2\pi\sigma^2}}e^{-\frac{(z-\mu)^2}{2\sigma^2}}, \label{eq:SU_phi_h}
\end{equation}
which can be interpreted as an approximation of the integral of $\phi(\,\cdot\,;\mu,\sigma^2)$ from $z$
to $z+h$ when $h$ is small. 

\subsection{Approximation for $\rcus(n,M)$} \label{sub:SU_PREFACTOR_S}

In the proof of Theorem \ref{thm:SA_IID_LIM} below, we derive an approximation
$\rcushat$ of $\rcus$ taking the form 
\begin{equation}
    \rcushat(n,M)\triangleq\alpha_n(Q,R,s)e^{-n(\Eziid(Q,\rhohat,s)-\rhohat R)},\label{eq:SA_RCU_s_hat}
\end{equation}
where $R=\frac{1}{n}\log M$, and the prefactor $\alpha_n$ varies
depending on whether $i_s(X,Y)$ is a lattice variable.  
In the non-lattice case, the prefactor is given by
\begin{multline}
    \alphanl(Q,R,s) \triangleq \int_{0}^{\infty}e^{-\rhohat z}\phi(z;nc_1,nc_2)dz \\ + \int_{-\infty}^{0}e^{(1-\rhohat)z}\phi(z;nc_1,nc_2)dz. \label{eq:SA_PreFactorIID}
\end{multline}
In the lattice
case, it will prove convenient to deal with $R-i_s(X,Y)$ rather than $i_s(X,Y)$.
Denoting the offset and span of $R-i_s(X,Y)$ by $\gamma$ and $h$ respectively, we 
see that $nR-i_s^n(\Xv,\Yv)$ has span $h$, and its offset can be chosen as
\begin{equation}
    \gamma_n \triangleq \min\Big\{ n\gamma+ih \,:\, i\in\ZZ, n\gamma+ih \ge 0 \Big\}. \label{eq:SU_alphan}
\end{equation}
The prefactor for the lattice case is given by
\begin{multline}
    \alphal(Q,R,s) \triangleq \sum_{i=0}^\infty e^{-\rhohat(\gamma_n+ih)}\phi_h(\gamma_n+ih;nc_1,nc_2) \\ + \sum_{i=-\infty}^{-1}e^{(1-\rhohat)(\gamma_n+ih)}\phi_h(\gamma_n+ih;nc_1,nc_2), \label{eq:SA_PreFactorIID_L}
\end{multline}
and the overall prefactor in \eqref{eq:SA_RCU_s_hat} is defined as
\begin{equation}
    \alpha_n \triangleq \begin{cases} \alphanl & i_s(X,Y)\text{ is non-lattice} \\ \alphal & R-i_s(X,Y) \text{ has offset $\gamma$ and span $h$}. \end{cases} 
\end{equation}
While \eqref{eq:SA_PreFactorIID} and \eqref{eq:SA_PreFactorIID_L} are written in
terms of integrals and summations, both prefactors can be computed 
efficiently to a high degree of accuracy.  In the non-lattice case, this is 
easily done using the identity
\begin{equation}
    \int_{a}^{\infty}e^{bz}\phi(z;\mu,\sigma^2)dz = e^{\mu b + \frac{1}{2}\sigma^2b^2}\Qsf\Big(\frac{a-\mu-b\sigma^2}{\sigma}\Big). \label{eq:SU_Integral}
\end{equation}
In the lattice case, we can write each of the summations in
\eqref{eq:SA_PreFactorIID_L} in the form
\begin{equation}
    \sum_i e^{b_0 + b_1i + b_2i^2} = e^{-\frac{b_1^2}{4b_2}+b_0}\sum_i e^{b_2(i+\frac{b_1}{2b_2})^2}, \label{eq:SA_SampleSum} \\
\end{equation}
where $b_2 < 0$.  We can thus obtain an accurate approximation by
keeping only the terms in the sum such that $i$ is sufficiently 
close to $-\frac{b_1}{2b_2}$.  Overall, the computational complexity
of the saddlepoint approximation is similar to that of the exponent alone.

\begin{thm} \label{thm:SA_IID_LIM}
    Consider a DMC $W$, decoding metric $q$, input distribution $Q$, and 
    parameter $s>0$ such that $U_s(Q)>0$.  For any sequence
    $\{M_n\}$ such that $M_n\to\infty$, we have 
    \begin{equation}
        \lim_{n\to\infty} \frac{\rcushat(n,M_n)}{\rcus(n,M_n)} = 1. \label{eq:SA_rcusResult}
    \end{equation}
\end{thm}
\begin{IEEEproof}
    See Appendix \ref{sub:SU_SADDLEPOINT_PROOF}.
\end{IEEEproof}

A heuristic derivation of the non-lattice version of $\rcushat$ was provided
in \cite{Saddlepoint};  Theorem \ref{thm:SA_IID_LIM} provides a formal
derivation, along with a treatment of the lattice case.
It should be noted that the assumption $U_s(Q)>0$ is not restrictive,
since in the case that $U_s(Q)=0$ the argument to the expectation in \eqref{eq:SU_RCU_s2}
is deterministic, and hence $\rcus$ can easily be computed exactly.

In the case that the rate $R$ is fixed, simpler asymptotic expressions
can be obtained. In Appendix \ref{sub:SU_ASYMPTOTIC_SA}, we prove the following 
(here $f_n \asymp g_n$ denotes the relation $\lim_{n\to\infty}\frac{f_n}{g_n}=1$):
\begin{itemize}
  \item If $R\in[0,\Rcrs(Q))$ or $R>I_s(Q)$, then
        \begin{equation}
            \alpha_n(Q,R,s) \asymp 1. \label{eq:SU_BelowRcr}
        \end{equation}
  \item If $R=\Rcrs(Q)$ or $R=I_s(Q)$, then
        \begin{equation}
            \alpha_n(Q,R,s) \asymp \frac{1}{2}.  \label{eq:SU_EqualRcr}
        \end{equation}
  \item If $R\in(\Rcrs(Q),I_s(Q))$, then
        \begin{align}
            \alphanl(Q,R,s)  &\asymp \frac{1}{\sqrt{2\pi nc_2}\rhohat(1-\rhohat)} \label{eq:SU_AboveRcr} \\
            \alphal(Q,R,s) &\asymp \frac{h}{\sqrt{2\pi nc_2}} \nonumber \\
                & \hspace*{-18mm} \times \Bigg(e^{-\rhohat\gamma_n}\bigg(\frac{1}{1-e^{-\rhohat h}}\bigg) + e^{(1-\rhohat)\gamma_n}\bigg(\frac{e^{-(1-\rhohat)h}}{1-e^{-(1-\rhohat)h}}\bigg)\Bigg). \label{eq:SU_AboveRcr_L}
        \end{align}
\end{itemize}
The asymptotic prefactors in \eqref{eq:SU_BelowRcr}--\eqref{eq:SU_AboveRcr_L}
are related to the problem of \emph{exact asymptotics} in the statistics 
literature, which seeks to characterize the subexponential prefactor for
probabilities that decay at an exponential rate (e.g. see \cite{BahadurRao}).
These prefactors are useful in gaining further insight into the behavior of the
error probability compared to the error exponent alone.  However, there
is a notable limitation which is best demonstrated here using \eqref{eq:SU_AboveRcr}.
The right-hand side of \eqref{eq:SU_AboveRcr} characterizes the 
prefactor to within a multiplicative $1+o(1)$ term for a given rate, but it diverges as
$\rhohat\to0$ or $\rhohat\to1$.  Thus, unless $n$ is large, the
estimate obtained by omitting the higher-order terms is inaccurate
for rates slightly above $\Rcrs(Q)$ or slightly below $I_s(Q)$.

In contrast, the right-hand side of \eqref{eq:SA_PreFactorIID} (and 
similarly \eqref{eq:SA_PreFactorIID_L}) remains bounded
for all $\rhohat\in[0,1]$.  Furthermore, as Theorem \ref{thm:SA_IID_LIM}
shows, this expression characterizes the true behavior of $\rcus$ to within a 
multiplicative $1+o(1)$ term not only for fixed rates, but also
when the rate varies with the block length.  Thus, it remains suitable
for characterizing the behavior of $\rcus$ even when the rate approaches
$\Rcrs(Q)$ or $I_s(Q)$.  In particular, this implies that $\rcushat$ 
gives the correct second-order asymptotics of the rate for a given
target error probability (see \eqref{eq:SU_DispersionIID}).  More
precisely, the proof of Theorem \ref{thm:SA_IID_LIM} reveals that
$\rcushat = \rcus + O\big(\frac{1}{\sqrt{n}}\big)$, which implies
(via a Taylor expansion of $\Qsf^{-1}$ in \eqref{eq:SU_DispersionIID})
that the two yield the same asymptotics for a given error probability
up to the $O(1)$ term.

\subsection{Approximation for $\rcu(n,M)$} \label{sub:SU_PREFACTOR_RCU}

In the proof of Theorem \ref{thm:SU_FiniteUB}, we obtained $\rcus$ from
$\rcu$ using Markov's inequality.  In this subsection we will see that, under
some technical assumptions, a more refined analysis yields a bound which
is tighter than $\rcus$, but still amenable to the techniques of the previous
subsection.   

\subsubsection{Technical Assumptions}

Defining the set
\begin{multline}
    \Yc_1(Q) \triangleq \Big\{ y\,:\, q(x,y) \ne q(\overline{x},y)\text{ for some } \\
         x,\xbar\text{ such that }Q(x)Q(\xbar)W(y|x)W(y|\xbar)>0\Big\}, \label{eq:REF_SetY1}
\end{multline}
the technical assumptions on $(W,q,Q)$ are as follows:
\begin{equation}
    q(x,y)>0 \iff W(y|x)>0 \label{eq:REF_Assumption1}
\end{equation}
\begin{equation}
    \Yc_1(Q) \ne \emptyset.\label{eq:REF_Assumption2}
\end{equation}
When $q(x,y)=W(y|x)$, \eqref{eq:REF_Assumption1} is trivial,
and \eqref{eq:REF_Assumption2} is the \emph{non-singularity}
condition of \cite{RefinementJournal}. A notable example where
this condition fails is the binary erasure channel (BEC) with $Q=\big(\frac{1}{2},\frac{1}{2}\big)$.
It should be noted that if \eqref{eq:REF_Assumption1} holds but 
\eqref{eq:REF_Assumption2} fails then we  in fact have $\rcu=\rcus$
for any $s>0$, and hence $\rcushat$ also approximates $\rcu$. 
This can be seen by noting that $\rcus$ is obtained from $\rcu$ using the inequality
$\openone\{\qbar \ge q \} \le \big(\frac{\qbar}{q}\big)^{s}$, which holds with equality
when $\frac{\qbar}{q}\in\{0,1\}$.

\subsubsection{Definitions}

Along with the definitions in Section \ref{sec:SU_REF_DEFS}, we will make 
use of the reverse conditional distribution 
\begin{align}
    \Ptilde_s(x|y) & \defeq\frac{Q(x)q(x,y)^{s}}{\sum_{\xbar}Q(\xbar)q(\xbar,y)^{s}} \label{eq:REF_Vs},
\end{align}
the joint tilted distribution  
\begin{equation}
    P_{\rhohat,s}^{*}(x,y)=\frac{Q(x)W(y|x)e^{-\rhohat i_s(x,y)}}{\sum_{x^{\prime},y^{\prime}}Q(x^{\prime})W(y^{\prime}|x^{\prime})e^{-\rhohat i_s(x',y')}}, \label{eq:SU_P*XY}
\end{equation}
and its $Y$-marginal $P_{\rhohat,s}^{*}(y)$, and the conditional variance 
\begin{equation}
    c_3(Q,R,s) \triangleq \EE\Big[ \var\big[i_s(X^{*}_s,Y_s^{*}) \big| Y_s^{*} \big] \Big], \label{eq:REF_sigma_s}
\end{equation}
where $(X_s^*,Y_s^*) \sim P_{\rhohat,s}^{*}(y)\Ptilde_s(x|y)$.  Furthermore, we define
\begin{equation}
    \Ic_s \triangleq \Big\{ i_s(x,y) \,:\, Q(x)W(y|x)>0,y\in\Yc_1(Q)\Big\}
\end{equation}
and let
\begin{equation}
    \psi_{s} \triangleq 
        \begin{cases}
            1 & \Ic_s\text{ does not lie on a lattice} \\
            \frac{\hover}{1-e^{-\hover}} & \Ic_s\text{ lies on a lattice with span }\hover.
        \end{cases} \label{eq:SA_Psi_s}
\end{equation}
The set $\Ic_s$ is the support of a random variable which
will appear in the analysis of the inner probability in \eqref{eq:SU_RCU}.
While $\hover$ in \eqref{eq:SA_Psi_s} can differ from $h$ (the span of 
$i_s(X,Y)$) in general, the two coincide whenever $\Yc_1(Q)=\Yc$.

We claim that the assumptions in
\eqref{eq:REF_Assumption1}--\eqref{eq:REF_Assumption2} imply that $c_3 > 0$
for any $R$ and $s>0$.  To see this, we write
\begin{align} 
     & \var_{\Ptilde_s(\cdot|y)}[i_{s}(X,y)]=0 \nonumber \\
     & \iff\log\frac{\Ptilde_s(x|y)}{Q(x)}\text{ is independent of } x\text{ where }\Ptilde_s(x|y)>0 \label{eq:REF_VarProof2}\\
     & \iff q(x,y)\text{ is independent of }x\text{ where }Q(x)q(x,y)>0 \label{eq:REF_VarProof3}\\
     & \iff y\notin\Yc_1(Q),\label{eq:REF_VarProof4}
\end{align}
where \eqref{eq:REF_VarProof2} and \eqref{eq:REF_VarProof3} follow from 
the definition of $\Ptilde_s$ in \eqref{eq:REF_Vs} and the
assumption $s>0$, and \eqref{eq:REF_VarProof4} follows from \eqref{eq:REF_Assumption1}
and the definition of $\Yc_1(Q)$ in \eqref{eq:REF_SetY1}.  Using \eqref{eq:SU_is},
\eqref{eq:REF_Assumption1} and \eqref{eq:SU_P*XY}, we have
\begin{equation}
    P_{\rhohat,s}^{*}(y) > 0 \iff \sum_{x}Q(x)W(y|x)>0.
\end{equation}
Thus, from \eqref{eq:REF_Assumption2}, we have $P^{*}_{\rhohat,s}(y) > 0$ 
for some $y \in \Yc_1(Q)$, which (along with \eqref{eq:REF_VarProof4}) proves that $c_3 > 0$.

\subsubsection{Main Result}

The main result of this subsection is written in terms of an approximation of the form
\begin{equation}
    \rcusshat(n,M)\triangleq\beta_n(Q,R,s)e^{-n(\Eziid(Q,\rhohat,s)-\rhohat R)}. \label{eq:SA_RCU_hat}
\end{equation}
Analogously to the previous subsection, 
we treat the lattice and non-lattice cases separately, writing
\begin{equation}
    \beta_n \triangleq \begin{cases} \betanl & i_s(X,Y)\text{ is non-lattice} \\ \betal & R-i_s(X,Y)\text{ has offset $\gamma$ and span $h$}, \end{cases} 
\end{equation}
where
\begin{align}
    &\betanl(Q,R,s) \triangleq \int_{\log\frac{\sqrt{2\pi nc_3}}{\psi_s}}^{\infty}e^{-\rhohat z}\phi(z;nc_1,nc_2)dz \nonumber \\
        & + \frac{\psi_s}{\sqrt{2\pi nc_3}}\int_{-\infty}^{\log\frac{\sqrt{2\pi nc_3}}{\psi_s}}e^{(1-\rhohat)z}\phi(z;nc_1,nc_2)dz \label{eq:SA_PreFactorRCU} \\
    &\betal(Q,R,s) \triangleq \sum_{i=i^{*}}^\infty e^{-\rhohat(\gamma_n+ih)}\phi_h(\gamma_n+ih;nc_1,nc_2) \nonumber \\
        &  + \frac{\psi_s}{\sqrt{2\pi nc_3}}\sum_{i=-\infty}^{i^{*}-1}e^{(1-\rhohat)(\gamma_n+ih)}\phi_h(\gamma_n+ih;nc_1,nc_2), \label{eq:SA_PreFactorRCU_L}
\end{align} 
and where in \eqref{eq:SA_PreFactorRCU_L} we use $\gamma_n$ in \eqref{eq:SU_alphan} along with
\begin{equation}
    i^{*} \triangleq \min\bigg\{ i\in\ZZ \,:\, \gamma_n + ih \ge \log\frac{\sqrt{2\pi nc_3}}{\psi_s} \bigg\}. \label{eq:SU_i*}
\end{equation}

\begin{thm} \label{thm:SA_SaddleRCU}
    Under the setup of Theorem \ref{thm:SA_IID_LIM} and the 
    assumptions in \eqref{eq:REF_Assumption1}--\eqref{eq:REF_Assumption2}, 
    we have for any $s>0$ that
    \begin{equation}
        \rcu(n,M_n) \le \rcuss(n,M_n)(1+o(1)), \label{eq:SU_PfRCU_1}
    \end{equation}
    where
    \begin{equation}
        \rcuss(n,M) \triangleq \EE\bigg[\min\bigg\{1,\frac{M\psi_s}{\sqrt{2\pi nc_3}}e^{-i_s^n(\Xv,\Yv)} \bigg\}\bigg]. \label{eq:SU_RCU_s_star}
    \end{equation}
    Furthermore, we have
    \begin{equation}
        \lim_{n\to\infty} \frac{\rcusshat(n,M_n)}{\rcuss(n,M_n)} = 1. \label{eq:SU_SaddleRCU}
    \end{equation} 
\end{thm}
\begin{IEEEproof}
    See Appendix \ref{sub:SA_RCU_PROOF}.
\end{IEEEproof}

When the rate does not vary with $n$, we can apply the same arguments
as those given in Appendix \ref{sub:SU_ASYMPTOTIC_SA} to obtain the 
following analogues of \eqref{eq:SU_BelowRcr}--\eqref{eq:SU_AboveRcr_L}:
\begin{itemize}
  \item If $R\in[0,\Rcrs(Q))$, then
        \begin{align}
            \beta_n(Q,R,s) &\asymp \frac{\psi_s}{\sqrt{2\pi nc_3}}, \label{eq:SU_BelowRcr2}
        \end{align}
        and similarly for $R=\Rcrs(Q)$ after multiplying the right-hand side by $\frac{1}{2}$.
  \item If $R\in(\Rcrs(Q),I_s(Q))$, then
        \begin{align}
            \betanl(Q,R,s)  &\asymp \bigg(\frac{\psi_s}{\sqrt{2\pi nc_3}}\bigg)^{\rhohat} \frac{1}{\sqrt{2\pi nc_2}\rhohat(1-\rhohat)} \label{eq:SU_AboveRcr2} \\
            \betal(Q,R,s) &\asymp \bigg(\frac{\psi_s}{\sqrt{2\pi nc_3}}\bigg)^{\rhohat} \frac{h}{\sqrt{2\pi nc_2}} \nonumber \\ 
                & \hspace*{-20mm} \times\Bigg(e^{-\rhohat\gamma'_n}\bigg(\frac{1}{1-e^{-\rhohat h}}\bigg)+e^{(1-\rhohat)\gamma'_n}\bigg(\frac{e^{-(1-\rhohat)h}}{1-e^{-(1-\rhohat)h}}\bigg)\Bigg), \label{eq:SU_AboveRcr_L2}
        \end{align}
        where $\gamma'_n \triangleq \gamma_n + i^{*}h - \log\frac{\sqrt{2\pi nc_3}}{\psi_s}\in[0,h)$ (see \eqref{eq:SU_i*}).
  \item For $R \ge I_s(Q)$, the asymptotics of $\beta_n$ coincide with 
        those of $\alpha_n$ (see \eqref{eq:SU_BelowRcr}--\eqref{eq:SU_EqualRcr}).
\end{itemize} 
When combined with Theorem \ref{thm:SA_SaddleRCU}, these expansions 
provide an alternative proof of the main result of \cite{RefinementJournal},
along with a characterization of the multiplicative $\Theta(1)$ terms which
were left unspecified in \cite{RefinementJournal}.  A simpler version
of the analysis in this paper can also be used to obtain the prefactors
with unspecified constants; see \cite{PaperRefinement} for details.

Analogously to the previous section, in the regime of fixed error 
probability we can write \eqref{eq:SU_SaddleRCU} more precisely as  
$\rcusshat = \rcuss + O\big(\frac{1}{\sqrt n}\big)$, implying that
the asymptotic expansions of the rates corresponding to $\rcuss$ and $\rcusshat$ 
coincide up to the $O(1)$ term.  From the analysis given in \cite[Sec. 3.4.5]{FiniteThesis},
$\rcuss$ yields an expansion of the form \eqref{eq:SU_DispersionIID} 
with the $O(1)$ term replaced by $\frac{1}{2}\log{n}+O(1)$. 
It follows that the same is true of $\rcusshat$.
 
\subsection{Numerical Examples} \label{sub:SA_NUMERICAL}

Here we provide numerical examples to demonstrate the utility of the 
saddlepoint approximations given in this section.  
Along with $\rcushat$ and $\rcusshat$, we consider (i) the normal
approximation, obtained by omitted the remainder term in \eqref{eq:SU_DispersionIID},
(ii) the error exponent approximation $\pe\approx e^{-n\Eriid(Q,R)}$, and (iii)
exact asymptotics approximations, obtained by ignoring the implicit 
$1+o(1)$ terms in \eqref{eq:SU_AboveRcr_L} and \eqref{eq:SU_AboveRcr_L2}.
We use the lattice-type versions of the approximations, since we consider examples
in which $i_s(X,Y)$ is a lattice variable.  We observed no significant difference
in the accuracy of each approximation in similar non-lattice examples.

We consider the example given in Section \ref{sub:SU_NUMERICAL},
using the parameters $\delta_{0}=0.01$, $\delta_{1}=0.05$, $\delta_{2}=0.25$,
and $Q=(\frac{1}{3},\frac{1}{3},\frac{1}{3})$. For the
saddlepoint approximations, we approximate the summations of the form \eqref{eq:SA_SampleSum}
by keeping the 1000 terms\footnote{The plots remained the same when this value was increased or decreased by an order of magnitude.} whose indices are closest to $-\frac{b_1}{2b_2}$.  
We choose the free parameter $s$
to be the value which maximizes the error exponent at each rate.  For the 
normal approximation, we choose to $s$ achieve the GMI in \eqref{eq:INTR_RateGMI}.
Defining $\Rcr(Q)$ to be the supremum of all rates such that $\rhohat=1$ when
$s$ is optimized, we have $\GMI(Q)=0.643$ and $\Rcr(Q)=0.185$ bits/use.

In Figure \ref{fig:SA_60_IID}, we plot the error probability
as a function of the rate with $n=60$. Despite the fact that 
the block length is small, we observe that $\rcus$
and $\rcusshat$ are indistinguishable at all rates.
Similarly, the gap from $\rcu$ to $\rcusshat$ is small. 
Consistent with the fact that Theorem \ref{thm:SA_SaddleRCU}
gives an asymptotic upper bound on $\rcu$ rather than an 
asymptotic equality, $\rcusshat$ lies slightly above $\rcu$ at 
low rates.
The error exponent approximation is close to $\rcus$
at low rates, but it is pessimistic at high rates. The normal approximation 
behaves somewhat similarly to $\rcus$, but it is less  
precise than the saddlepoint approximation, particularly at low rates.
 
\begin{figure}
    \begin{centering}
        \includegraphics[width=1\columnwidth]{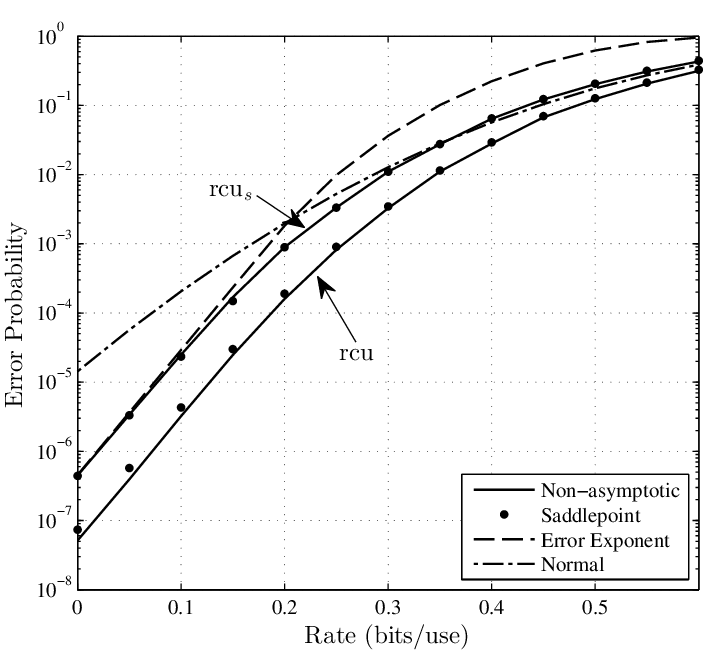} 
        \par
    \end{centering}
 
    \caption{i.i.d. random-coding bounds for the channel
             defined in \eqref{eq:SU_MatrixW} with minimum Hamming distance decoding.
             The parameters are $n=60$, $\delta_{0}=0.01$, $\delta_{1}=0.05$, $\delta_{2}=0.25$
             and $Q=(\frac{1}{3},\frac{1}{3},\frac{1}{3})$. \label{fig:SA_60_IID}}
\end{figure}

\begin{figure}
    \begin{centering}
        \includegraphics[width=1\columnwidth]{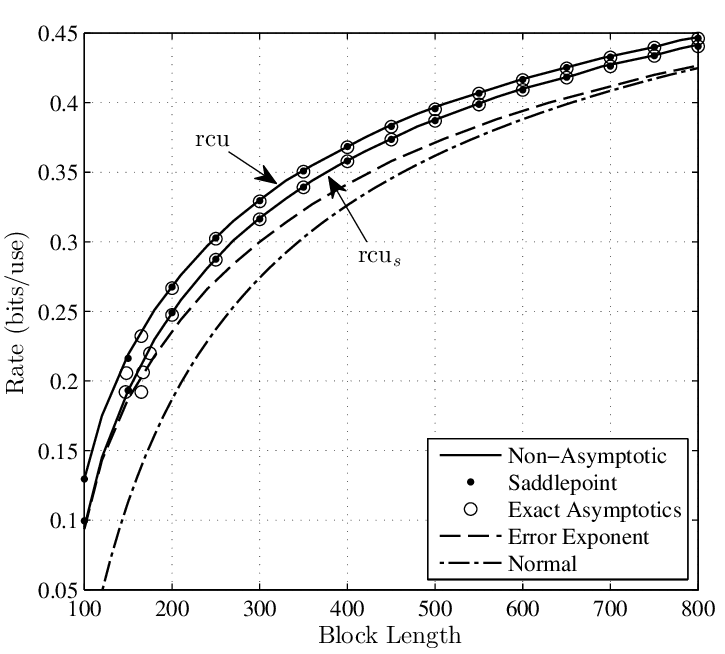} 
        \par
    \end{centering}
 
    \caption{Rate required to achieve a target error probability $\epsilon$ for
        the channel defined in \eqref{eq:SU_MatrixW} with ML decoding. The parameters are $\epsilon=10^{-8}$,
        $\delta_{0}=\delta_{1}=\delta_{2}=\delta=0.1$ and $Q=(\frac{1}{3},\frac{1}{3},\frac{1}{3})$.
        \label{fig:SA_IID_R}}
\end{figure}

To facilitate the computation of $\rcu$ and $\rcus$
at larger block lengths, we consider the symmetric setup of $\delta_{0}=\delta_{1}=\delta_{2}=\delta=0.1$
and $Q=\big(\frac{1}{3},\frac{1}{3},\frac{1}{3}\big)$.
Under these parameters, we have $I(X;Y)=0.633$ and $\Rcr(Q)=0.192$
bits/use. In Figure \ref{fig:SA_IID_R}, we plot the rate required
for each random-coding bound and approximation to achieve a given
error probability $\epsilon=10^{-8}$, as a function of $n$. Once
again, $\rcushat$ is indistinguishable from $\rcus$, 
and similarly for $\rcusshat$ and $\rcu$.
The error exponent approximation yields similar 
behavior to $\rcus$ at small block lengths, but the gap widens at
larger block lengths.  The exact asymptotics approximations are
accurate other than a divergence near the critical rate, which
is to be expected from the discussion in Section \ref{sub:SU_PREFACTOR_S}.
In contrast to similar plots with larger target error probabilities
(e.g. \cite[Fig. 8]{Finite}), the normal approximation 
is inaccurate over a wide range of rates.

\section{Discussion and Conclusion} \label{sec:CONCLUSION}

We have introduced a cost-constrained ensemble with multiple auxiliary
costs which yields similar performance gains to constant-composition coding, 
while remaining applicable in the case of infinite or continuous alphabets.  
We have studied the number of auxiliary costs required to match the performance
of the constant-composition ensemble, and shown that the number can be reduced 
when the input distribution or decoding metric is optimized.  Using
the saddlepoint approximation, refined asymptotic estimates have been given
for the i.i.d. ensemble which unify the regimes of error exponents,
second-order rates and moderate deviations, and provide accurate 
approximations of the random-coding bounds.

\subsection*{Extension to Channels with Input Constraints}

Suppose that each codeword $\xv$ is constrained
to satisfy $\frac{1}{n}\sum_{i=1}^{n}c(x_{i})\le\Gamma$
for some (system) cost function $c(\cdot)$. The i.i.d. ensemble is
no longer suitable, since in all non-trivial cases it has a positive
probability of producing codewords which violate the constraint. On
the other hand, the results for the constant-composition ensemble
remain unchanged provided that $Q$ itself satisfies the cost constraint,
i.e. $\sum_{x}Q(x)c(x)\le\Gamma$.

For the cost-constrained ensemble, the extension is less trivial but
still straightforward. The main change required is a modification
of the definition of $\Dc_{n}$ in \eqref{eq:SU_SetDn} to
include a constraint on the quantity $\frac{1}{n}\sum_{i=1}^{n}c(x_{i})$.
Unlike the auxiliary costs in \eqref{eq:SU_SetDn}, where the sample
mean can be above or below the true mean, the system cost of
each codeword is constrained to be less than or equal to its mean.
That is, the additional constraint is given by
\begin{equation}
    \frac{1}{n}\sum_{i=1}^{n}c(x_{i})\le\phi_{c}\triangleq\sum_{x}Q(x)c(x),
\end{equation}
or similarly with both upper and lower bounds (e.g. $-\frac{\delta}{n}\le\frac{1}{n}\sum_{i=1}^{n}c(x_{i})-\phi_{c}\le0$). 
Using this modified definition of $\Dc_{n}$,
one can prove the subexponential behavior of $\mu_{n}$ in Proposition
\ref{prop:SU_SubExpCost} provided that $Q$ is such that $\phi_{c}\le\Gamma$, and the 
exponents and second-order rates for the cost-constrained ensemble remain valid 
under any such $Q$.

\appendix

\subsection{Proof of Theorem \ref{thm:SU_EnsembleExps}} \label{sub:SU_PRIMAL_PROOF}

The proof is similar to that of Gallager for the constant-composition ensemble \cite{GallagerCC},
so we omit some details.  The codeword distribution in \eqref{eq:SU_Q_Multi} can be written as 
\begin{equation}
    P_{\Xv}(\xv)=\frac{1}{\mu_{n}}\prod_{i=1}^{n}Q(x_{i})\openone\big\{\hat{P}_{\xv}\in\Gc_{n}\big\},\label{eq:SU_Q_General}
\end{equation}
where $\hat{P}_{\xv}$ is the empirical distribution (type) of $\xv$, and
$\Gc_{n}$ is the set of types corresponding to sequences $\xv\in\Dc_n$ 
(see \eqref{eq:SU_SetDn}). We define the sets 
\begin{equation}
    \Sc_{n}(\Gc_{n})\triangleq\big\{ P_{XY}\in\Pc_{n}(\Xc\times\Yc)\,:\, P_{X}\in\Gc_{n}\big\}\label{eq:SU_SetS_Gen}
\end{equation}
\begin{multline}
    \mathcal{T}_{n}(P_{XY},\Gc_{n})\triangleq\big\{\Ptilde_{XY}\in\Pc_{n}(\Xc\times\Yc)\,:\, \Ptilde_{X}\in\Gc_{n}, \\ \Ptilde_{Y}=P_{Y},\EE_{\Ptilde}[\log q(X,Y)]\ge\EE_{P}[\log q(X,Y)]\big\}.\label{eq:SU_SetT_Gen}
\end{multline}
We have from Theorem \ref{thm:SU_FiniteUB} that $\pebar\doteq\rcu$.
Expanding $\rcu$ in terms of types, we obtain 
\begin{multline}
    \pebar\doteq\sum_{P_{XY}\in\Sc_{n}(\Gc_{n})}\PP\big[(\Xv,\Yv)\in T^{n}(P_{XY})\big]\min\bigg\{1, \\
    (M-1)\sum_{\Ptilde_{XY}\in\mathcal{T}_{n}(P_{XY},\Gc_{n})}\PP\big[(\Xvbar,\yv)\in T^{n}(\Ptilde_{XY})\big]\bigg\},\label{eq:SU_ThmGen_6}
\end{multline}
where $\yv$ denotes an arbitrary sequence with type $P_{Y}$.

From Proposition \ref{prop:SU_SubExpCost}, the normalizing constant
in \eqref{eq:SU_Q_General} satisfies $\mu_{n}\doteq1$, and thus
we can safely proceed from \eqref{eq:SU_ThmGen_6} as if the codeword
distribution were $P_{\Xv}=Q^n$. 
Using the property of types in \cite[Eq. (18)]{GallagerCC}, it follows that
the two probabilities in \eqref{eq:SU_ThmGen_6} behave as
$e^{-n D(P_{XY}\|Q\times W)}$ and $e^{-n D(\Ptilde_{XY}\|Q\times P_{Y})}$
respectively. Combining this with the fact that the number 
of joint types is polynomial in $n$, we obtain
$\pebar\doteq e^{-nE_{r,n}(Q,R,\Gc_{n})}\label{eq:SU_pe_Gen}$, where 
\begin{multline}
    E_{r,n}(Q,R,\Gc_{n})\triangleq\min_{P_{XY}\in\Sc_{n}(\Gc_{n})}\min_{\Ptilde_{XY}\in\mathcal{T}_{n}(P_{XY},\Gc_{n})} \\ D(P_{XY}\|Q\times W)+\big[D(\Ptilde_{XY}\|Q\times P_{Y})-R\big]^{+}.\label{eq:SU_Er_Gen}
\end{multline}
Using a simple continuity argument (e.g. see \cite[Eq. (30)]{DyachkovCC}), we 
can replace the minimizations over types by minimizations over joint distributions,
and the constraints of the form $|\EE_{P}[a_{l}(X)]-\phi_{l}|\le\frac{\delta}{n}$
can be replaced by $\EE_{P}[a_{l}(X)]=\phi_{l}$.  This concludes the proof.

\subsection{Proof of Theorem \ref{thm:SU_LagrangeDual}} \label{sub:SU_LAGRANGE_PROOF}

Throughout the proof, we make use of Fan's minimax theorem \cite{Minimax}, which
states that $\min_a \sup_b f(a,b) = \sup_b \min_a f(a,b)$ provided that
the minimum is over a compact set, $f(\cdot,b)$ is convex in $a$ for all $b$, 
and $f(a,\cdot)$ is concave in $b$ for all $a$.  We make use of Lagrange
duality \cite{Convex} in a similar fashion to \cite[Appendix A]{Merhav};
some details are omitted to avoid repetition with \cite{Merhav}.

Using the identity $[\alpha]^{+}=\max_{\rho\in[0,1]}\rho\alpha$ 
and Fan's minimax theorem, the expression in \eqref{eq:SU_Er_Cost} can be written as 
\begin{equation}
    \Ercc(Q,R)=\max_{\rho\in[0,1]}\Eohatcc(Q,\rho)-\rho R, \label{eq:SU_Er_Alt1}
\end{equation}
where
\begin{multline}
    \Eohatcc(Q,\rho)\triangleq\min_{P_{XY}\in\SetScost(\{a_l\})}\min_{\Ptilde_{XY}\in\SetTcost(P_{XY},\{a_l\})} \\ D(P_{XY}\|Q\times W)+\rho D(\Ptilde_{XY} \| Q \times P_Y). \label{eq:SU_Er_Alt2}
\end{multline} 
It remains to show that $\Eohatcc(Q,\rho)=\Ezcc(Q,\rho)$.
We will show this by considering the minimizations in \eqref{eq:SU_Er_Alt2}
one at a time. We can follow the steps of 
\cite[Appendix A]{Merhav} to conclude that
\begin{multline}
    \min_{\Ptilde_{XY}\in\SetTcost(P_{XY},\{a_l\})} D(\Ptilde_{XY} \| Q \times P_Y)  = \sup_{s\ge0,\{\rbar_l\}} \\ \sum_{x,y}P_{XY}(x,y)\log\frac{q(x,y)^s}{\sum_{\xbar}Q(\xbar)q(\xbar,y)^s e^{\sum_l \rbar_l(a_l(\xbar)-\phi_l)}}, \label{eq:SU_DualProof1_LM}
\end{multline}
where $s$ and $\{\rbar_l\}$ are Lagrange multipliers.
It follows that the inner minimization in \eqref{eq:SU_Er_Alt2} is equivalent to 
\begin{multline}
    \min_{P_{XY}\in\SetScost(\{a_l\})}\sup_{s\ge0,\{\rbar_l\}} \sum_{x,y}P_{XY}(x,y)\Bigg(\log\frac{P_{XY}(x,y)}{Q(x)W(y|x)} \\ + \rho\log\frac{q(x,y)^s}{\sum_{\xbar}Q(\xbar)q(\xbar,y)^s e^{\sum_l \rbar_l(a_l(\xbar)-\phi_l)}}\Bigg).\label{eq:SU_DualProof9_LM}
\end{multline}
Since the objective is convex in $P_{XY}$ and jointly concave in  
$(s,\{\rbar_l\})$, we can apply Fan's minimax theorem. 
Hence, we consider the minimization of the objective in \eqref{eq:SU_DualProof9_LM} over 
$P_{XY}\in\SetScost(\{a_l\})$ with $s$ and $\{\rbar_l\}$ fixed.  Applying
the techniques of \cite[Appendix A]{Merhav} a second time, we conclude   
that this minimization has a dual form given by
\begin{equation}
    \sup_{\{r_l\}} -\log\sum_{x,y}P_{XY}(x,y)\Bigg(\frac{\sum_{\xbar}Q(\xbar)q(\xbar,y)^s e^{\sum_l \rbar_l(a_l(\xbar)-\phi_l)}}{q(x,y)^s e^{\sum_l r_l(a_l(x)-\phi_l)}}\Bigg)^{\rho},
\end{equation}
where $\{r_l\}$ are Lagrange multipliers. The proof is concluded by 
taking the supremum over $s$ and $\{\rbar_l\}$.

\subsection{Necessary Conditions for the Optimal Parameters} \label{sub:SU_NECC_CONDS}

\subsubsection{Optimization of $\Ezcc(Q,\rho)$}

We write the objective in \eqref{eq:SU_E0_LM} as 
\begin{equation}  
    \Ezcc(Q,\rho,s,a) \triangleq \rho\sum_{x}Q(x)a(x)-\sum_{x}Q(x)\log f(x),
\end{equation}
where
\begin{equation}
    f(x)\defeq\sum_{y}W(y|x)q(x,y)^{-\rho s}\bigg(\sum_{\xbar}Q(\xbar)q(\xbar,y)^{s}e^{a(\xbar)}\bigg)^{\rho}. \label{eq:SU_def_f}
\end{equation}
We have the partial derivatives 
\begin{align}
    \frac{\partial f(x)}{\partial Q(x^{\prime})} & = \rho g(x,x^{\prime}) \\
    \frac{\partial f(x)}{\partial a(x^{\prime})} & =\rho Q(x^{\prime})g(x,x^{\prime}),
\end{align}
where
\begin{multline}
    g(x,x') \triangleq \sum_{y}W(y|x)q(x,y)^{-\rho s} \\ \times\rho\left(\sum_{\xbar}Q(\xbar)q(\xbar,y)^{s}e^{a(\xbar)}\right)^{\rho-1}q(x^{\prime},y)^{s}e^{a(x^{\prime})}
\end{multline}
We proceed by analyzing the necessary Karush-Kuhn-Tucker (KKT) conditions \cite{Convex} for $(Q,a)$ to 
maximize $\Ezcc(Q,\rho,s,a)$. The KKT condition corresponding to the partial 
derivative with respect to $a(x^{\prime})$ is
\begin{equation}
    \rho Q(x^{\prime})-\sum_{x}Q(x)\frac{\rho Q(x^{\prime})g(x,x^{\prime})}{f(x)}=0,
\end{equation}
or equivalently 
\begin{equation}
    \sum_{x}Q(x)\frac{g(x,x^{\prime})}{f(x)}=1.\label{eq:KKTa}
\end{equation}
Similarly, the KKT condition corresponding to $Q(x^{\prime})$ gives
\begin{equation}
    \rho a(x^{\prime})-\log f(x^{\prime})-\rho\sum_{x}Q(x)\frac{g(x,x^{\prime})}{f(x)}-\lambda=0 \label{eq:KKTQ1}
\end{equation}
for all $x'$ such that $Q(x')>0$, where $\lambda$ is the 
Lagrange multiplier associated with the constraint
$\sum_{x}Q(x)=1$. Substituting \eqref{eq:KKTa} into \eqref{eq:KKTQ1} gives
\begin{equation}
    -\log\big(f(x^{\prime})e^{-\rho a(x')}\big)=\lambda+\rho. \label{eq:SU_MaxQFinal1}
\end{equation}
Using the definition of $f(\cdot)$ in \eqref{eq:SU_def_f}, we see
that \eqref{eq:SU_MaxQFinal1} implies that the logarithm
in \eqref{eq:SU_E0_LM} is independent of $x$.

\subsubsection{Optimization of $I_{s,a}(Q)$}

We write $I_{s,a}(Q)$ in \eqref{eq:SU_Isa} as
\begin{multline}
    I_{s,a}(Q) = s\sum_{x,y}Q(x)W(y|x)\log q(x,y)+\sum_{x}Q(x)a(x) \\ -\sum_{x,y}Q(x)W(y|x)\log\sum_{\xbar}Q(\xbar)q(\xbar,y)^{s}e^{a(\xbar)}
\end{multline}
and analyze the KKT conditions associated with the maximization over
$(Q,a)$.  We omit some details, since the steps are similar to those above.
The KKT condition for $a(x^{\prime})$ is
\begin{equation}
    \sum_{x,y}Q(x)W(y|x)\frac{q(x^{\prime},y)^{s}e^{a(x^{\prime})}}{\sum_{\xbar}Q(\xbar)q(\xbar,y)^{s}e^{a(\xbar)}}=1, \label{eq:KKTa'}
\end{equation}
and the KKT condition for $Q(x^{\prime})$ gives
\begin{align}
    & s\sum_{y}W(y|x^{\prime})\log q(x^{\prime},y)+a(x^{\prime}) \nonumber \\ 
    & -\sum_{y}W(y|x^{\prime})\log\sum_{\xbar}Q(\xbar)q(\xbar,y)^{s}e^{a(\xbar)} \nonumber \\
    & -\sum_{x,y}Q(x)W(y|x)\frac{q(x^{\prime},y)^{s}e^{a(x^{\prime})}}{\sum_{\xbar}Q(\xbar)q(\xbar,y)^{s}e^{a(\xbar)}}-\lambda=0 \label{eq:KKTQ1'}
\end{align}
for all $x'$ such that $Q(x')>0$,
where $\lambda$ is a Lagrange multiplier.  Substituting \eqref{eq:KKTa'} 
into \eqref{eq:KKTQ1'} and performing some simple rearrangements, we obtain
\begin{equation}
    \sum_{y}W(y|x^{\prime})\log\frac{q(x^{\prime},y)^{s}e^{a(x^{\prime})}}{\sum_{\xbar}Q(\xbar)q(\xbar,y)^{s}e^{a(\xbar)}}=\lambda+1.
\end{equation} 

\subsection{Asymptotic Behavior of the Saddlepoint Approximation} \label{sub:SU_ASYMPTOTIC_SA}

Here we prove the asymptotic relations given in \eqref{eq:SU_BelowRcr}--\eqref{eq:SU_AboveRcr_L}.
We will make frequent use of the properties of $\rhohat$, $c_1$ and $c_2$
given in Section \ref{sec:SU_REF_DEFS}.

We first prove \eqref{eq:SU_BelowRcr}--\eqref{eq:SU_EqualRcr} in the non-lattice case.
Suppose that $R<\Rcrs(Q)$, and hence $\rhohat=1$, $c_1<0$ and $c_2>0$. 
Using \eqref{eq:SU_Integral} and the identity $\Qsf(z)\le\frac{1}{2}\exp\big(\frac{-z^2}{2}\big)$
for $z>0$, it is easily verified that the first term in \eqref{eq:SA_PreFactorIID} 
decays to zero exponentially fast. The second term is given by 
$\Qsf\big(-c_1\sqrt{\frac{n}{c_2}}\big)$, which tends to one since $\lim_{z\to\infty}\Qsf(z)=1$.  
We thus obtain \eqref{eq:SU_BelowRcr}. When $R=\Rcrs(Q)$, the argument is similar except that $c_1=0$,
yielding the following: (i) From \eqref{eq:SU_Integral}, the first term in \eqref{eq:SA_PreFactorIID} equals
$\Qsf(\sqrt{nc_2})e^{nc_2/2} \asymp \frac{1}{\sqrt{2\pi nc_2}}$, rather than decaying exponentially fast, 
(ii) The second term in \eqref{eq:SA_PreFactorIID} equals $\Qsf(0)=\frac{1}{2}$, rather than one.  
For $R > I_s(Q)$ (respectively, $R=I_s(Q)$) the argument 
is similar with the roles of the two terms in 
\eqref{eq:SA_PreFactorIID} reversed, and with $\rhohat=0$
and $c_1>0$ (respectively, $c_1=0$).

In the lattice case, the arguments in proving \eqref{eq:SU_BelowRcr}--\eqref{eq:SU_EqualRcr}
are similar to the non-lattice case, so we focus on \eqref{eq:SU_BelowRcr} with $R<\Rcrs(Q)$.  
Similarly to the non-lattice case, it is easily shown that the first
summation in \eqref{eq:SA_PreFactorIID_L} decays to zero exponentially
fast, so we focus on the second.  Since $\rhohat=1$, the second summation is given by
\begin{align}
    & \sum_{i=-\infty}^{-1} \phi_h(\gamma_n+ih;nc_1,nc_2) \nonumber \\
    & \qquad= (1+o(1))\sum_{i=-\infty}^{\infty}\phi_h(\gamma_n+ih;nc_1,nc_2) \label{eq:SU_LatticeAsymp4}\\
    & \qquad= 1+o(1), \label{eq:SU_LatticeAsymp4a} 
\end{align}
where \eqref{eq:SU_LatticeAsymp4} follows since the added terms from $i=0$
to $\infty$ contribute an exponentially small amount to the sum
since $c_1<0$, and  
\eqref{eq:SU_LatticeAsymp4a} is easily understood by interpreting
the right-hand side of \eqref{eq:SU_LatticeAsymp4} as approximating
the integral over the real line of a Gaussian density function via
discrete sampling.  Since the sampling is done using intervals of a fixed 
size $h$ but the variance $nc_2$ increases, the approximation 
improves with $n$ and approaches one.

Finally, we consider the case that $R\in(\Rcrs(Q),I_s(Q))$, and hence
$\rhohat\in(0,1)$, $c_1=0$ and $c_2>0$.  In the non-lattice case,
we can substitute $c_1=0$ and \eqref{eq:SU_Integral} into \eqref{eq:SA_PreFactorIID_L} to obtain
\begin{align}
    \alpha_n &= e^{\frac{1}{2}nc_2\rhohat^2}\Qsf\big(\rhohat\sqrt{nc_2}\big) + e^{\frac{1}{2}nc_2(1-\rhohat)^2}\Qsf\big((1-\rhohat)\sqrt{nc_2}\big). \label{eq:SU_LatticeAsymp5}
\end{align}
Using the fact that $\Qsf(z)e^{z^2/2}\asymp\frac{1}{z\sqrt{2\pi}}$ as $z\to\infty$, along with the 
identity $\frac{1}{\rho}+\frac{1}{1-\rho}=\frac{1}{\rho(1-\rho)}$,
we obtain \eqref{eq:SU_AboveRcr}.

We now turn to the lattice case. Setting $c_1=0$ in \eqref{eq:SA_PreFactorIID_L} yields
\begin{multline}
    \alpha_n = \frac{h}{\sqrt{2\pi nc_2}}\bigg( \sum_{i=0}^\infty e^{-\rhohat(\gamma_n+ih)-\frac{(\gamma_n+ih)^2}{2nc_2}} + \\ \sum_{i=-\infty}^{-1}e^{(1-\rhohat)(\gamma_n+ih)-\frac{(\gamma_n+ih)^2}{2nc_2}} \bigg). \label{eq:SU_LatticeAsymp6} 
\end{multline}
The two summations are handled in a nearly identical fashion, so
we focus on the first.  Using the identity $1-x \le e^{-x} \le 1$,
we can write 
\begin{align}
    & \Bigg| \sum_{i=0}^{\infty}e^{-\rhohat(\gamma_n+ih)-\frac{(\gamma_n+ih)^2}{2nc_2}} - \sum_{i=0}^{\infty}e^{-\rhohat(\gamma_n+ih)}\Bigg| \nonumber \\  
    & \qquad \le \sum_{i=0}^{\infty}e^{-\rhohat(\gamma_n+ih)}\frac{(\gamma_n+ih)^2}{2nc_2} \\
    & \qquad = O\Big(\frac{1}{n}\Big), \label{eq:SU_LatticeAsymp8}
\end{align}
where \eqref{eq:SU_LatticeAsymp8} follows since the summation $\sum_{i=0}^{\infty}e^{-\zeta i}p(i)$
is convergent for any polynomial $p(i)$ and $\zeta>0$. Furthermore, we have from the geometric series that
\begin{equation}
    \sum_{i=0}^{\infty}e^{-\rhohat(\gamma_n+ih)} = e^{-\rhohat\gamma_n}\bigg(\frac{1}{1-e^{-\rhohat h}}\bigg),
\end{equation}
We have thus weakened the first summation in \eqref{eq:SU_LatticeAsymp6}
to the first term in the sum in \eqref{eq:SU_AboveRcr_L} (up to an 
$O\big(\frac{1}{n}\big)$ remainder term). The second term is obtained in a similar fashion.

\subsection{Proof of Theorem \ref{thm:SA_IID_LIM}} \label{sub:SU_SADDLEPOINT_PROOF}

Since $M_n\to\infty$ by assumption, we can safely replace $M_n$ by $M_n+1$ without
affecting the theorem statement. We begin by considering fixed values of $n$, $M$
and $R=\frac{1}{n}\log M$.

Using \eqref{eq:SU_RCU_s2} and the identity in \eqref{eq:SU_UnifIneq}, we can can write
\begin{equation}
    \rcus(n,M+1)=\PP\bigg[nR-\sum_{i=1}^{n}i_{s}(X_{i},Y_{i})\ge\log U\bigg].\label{eq:SA_Above_1}
\end{equation}
This expression resembles the tail probability of an i.i.d.
sum of random variables, for which asymptotic estimates were given
by Bahadur and Rao \cite{BahadurRao} (see also \cite[Appendix 5A]{Gallager}).  
There are two notable differences
in our setting which mean that the results of \cite{BahadurRao,Gallager}
cannot be applied directly.  First, the right-hand side of the event in 
\eqref{eq:SA_Above_1} random rather than deterministic.  Second, since we are allowing
for rates below $\Rcrs(Q)$ or above $I_s(Q)$, we cannot assume that
the derivative of the moment generating function of $R-i_s(X,Y)$ at zero
(which we will shortly see equals $c_1$ in \eqref{eq:SA_c1}) is equal to zero.

\subsubsection{Alternative Expressions for $\rcus$} \label{sub:SA_ABOVE_RCR}

Let $F(t)$ denote the cumulative distribution function (CDF) of $R-i_{s}(X,Y)$,
and let $Z_{1},\cdots,Z_{n}$ be i.i.d. according to the tilted CDF
\begin{equation}
    F_{Z}(z)=e^{\Eziid(Q,\rhohat,s)-\rhohat R}\int_{-\infty}^{z}e^{\rhohat t}dF(t). \label{eq:SA_Above_2}
\end{equation}
It is easily seen that this is indeed a CDF by writing 
\begin{equation} 
    \int_{-\infty}^{\infty}e^{\rhohat t}dF(t)=\mathbb{E}\big[e^{\rhohat(R-i_{s}(X,Y))}\big]=e^{-(\Eziid(Q,\rhohat,s)-\rhohat R)},
\end{equation}
where we have used \eqref{eq:SU_E0s_IID}.  The moment generating function 
(MGF) of $Z$ is given by
\begin{align}
    M_Z(\tau) &\triangleq \EE\big[e^{\tau Z}\big] \label{eq:SU_MGF_Z} \\
              &= e^{\Eziid(Q,\rhohat,s)-\rhohat R}\EE\big[e^{(\rhohat+\tau)(R-i_{s}(X,Y))}\big] \label{eq:SA_Moments2} \\
              & = e^{\Eziid(Q,\rhohat,s)}e^{-(\Eziid(Q,\rhohat+\tau,s)- \tau R)}, \label{eq:SA_Moments3}
\end{align}
where \eqref{eq:SA_Moments2} follows from \eqref{eq:SA_Above_2}, and 
\eqref{eq:SA_Moments3} follows from \eqref{eq:SU_E0s_IID}.  We can now
compute the mean and variance of $Z$ in terms of the derivatives 
of the MGF, namely
\begin{align}
    \EE[Z] &= \frac{dM_Z}{d\tau}\Big|_{\tau=0} = c_1 \label{eq:SU_MGFder1} \\
    \var[Z] &= \frac{d^2M_Z}{d\tau^2}\Big|_{\tau=0} - \EE[Z]^2 = c_2, \label{eq:SU_MGFder2} 
\end{align}
where $c_1$ and $c_2$ are defined in \eqref{eq:SA_c1}--\eqref{eq:SA_c2}.
Recall that $U_s(Q)>0$ by assumption, which implies that $c_2>0$ (see Section \ref{sec:SU_REF_DEFS}).

In the remainder of the proof, we omit the arguments
$(Q,\rhohat,s)$ to $\Eziid$. Following \cite[Lemma 2]{BahadurRao}, we 
can use \eqref{eq:SA_Above_2} to write \eqref{eq:SA_Above_1} as follows: 
\begin{align}
     & \rcus(n,M+1) \nonumber \\
     & \quad = \dotsint_{\sum_{i}t_{i}\ge\log u}dF(t_{1})\cdots dF(t_{n})dF_{U}(u) \\
     & \quad = e^{-n(\Eziid-\rhohat R)} \nonumber \\ 
        &\quad\,\,\, \times \dotsint_{\sum_{i}z_{i}\ge\log u}e^{-\rhohat\sum_{i}z_{i}}dF_{Z}(z_{1})\cdots dF_{Z}(z_{n})dF_{U}(u),\label{eq:SA_ChgMeasure1} \\
     & \quad \triangleq I_n e^{-n(\Eziid-\rhohat R)}, \label{eq:SA_Above_3} 
\end{align}
where $F_{U}(u)$ is the CDF of $U$.  We write the prefactor $I_n$ as 
\begin{equation}
    I_n = \int_0^1\int_{\log u}^\infty e^{-\rhohat z} dF_{n}(z) dF_{U}(u), \label{eq:SA_ChgMeasure2}
\end{equation}   
where $F_n$ is the CDF of $\sum_{i=1}^n Z_i$.  Since the integrand
in \eqref{eq:SA_ChgMeasure2} is non-negative, we can safely interchange
the order of integration, yielding 
\begin{align}
    I_n &= \int_{-\infty}^{\infty} \int_0^{\min\{1,e^z\}} e^{-\rhohat z} dF_U(u) dF_n(z) \label{eq:SA_ChgMeasure4} \\
        &= \int_{0}^{\infty} e^{-\rhohat z} dF_n(z) + \int_{-\infty}^{0} e^{(1-\rhohat)z} dF_n(z), \label{eq:SA_ChgMeasure5a} 
\end{align} 
where \eqref{eq:SA_ChgMeasure5a} follows by splitting the integral according
to which value achieves the $\min\{\cdot,\cdot\}$ in \eqref{eq:SA_ChgMeasure4}.
Letting $\hat{F_n}$ denote the CDF of $\frac{\sum_{i=1}^n Z_i - nc_1}{\sqrt{nc_2}}$,
we can write \eqref{eq:SA_ChgMeasure5a} as
\begin{multline}
    I_n = \int_{-\frac{c_1\sqrt{n}}{\sqrt{c_2}}}^{\infty} e^{-\rhohat(z\sqrt{nc_2}+nc_1)} d\hat{F}_{n}(z) \\
    + \int_{-\infty}^{-\frac{c_1\sqrt{n}}{\sqrt{c_2}}} e^{(1-\rhohat)(z\sqrt{nc_2}+nc_1)} d\hat{F}_{n}(z). \label{eq:SA_ChgMeasure5} 
\end{multline} 
 
\subsubsection{Non-lattice Case}

Let $\Phi(z)$ denote the CDF of a zero-mean unit-variance Gaussian
random variable. Using the fact that $\EE[Z]=c_1$ and $\var[Z]=c_2>0$
(see \eqref{eq:SU_MGFder1}--\eqref{eq:SU_MGFder2}), we have from 
the refined central limit theorem in \cite[Sec. XVI.4, Thm. 1]{Feller} that
\begin{equation}
    \hat{F}_{n}(z)=\Phi(z)+G_{n}(z)+\tilde{F}_{n}(z),\label{eq:SA_Above_6}
\end{equation}
where $\tilde{F}_{n}(z)=o(n^{-\frac{1}{2}})$ uniformly in $z$, and
\begin{equation}
    G_{n}(z)\triangleq\frac{K}{\sqrt{n}}(1-z^{2})e^{-\frac{1}{2}z^{2}}\label{eq:SA_Above_7}
\end{equation}
for some constant $K$ depending only on the variance and third absolute moment
of $Z$, the latter of which is finite since we are considering
finite alphabets. Substituting \eqref{eq:SA_Above_6}
into \eqref{eq:SA_ChgMeasure5}, we obtain
\begin{equation}
    I_{n}=I_{1,n}+I_{2,n}+I_{3,n},\label{eq:SA_Above_8} 
\end{equation}
where the three terms denote the right-hand side of \eqref{eq:SA_ChgMeasure5}
with $\Phi$, $G_n$ and $\tilde{F}_n$ respectively in place of $\hat{F}_n$.  
Reversing the step from \eqref{eq:SA_ChgMeasure5a} to \eqref{eq:SA_ChgMeasure5},
we see that $I_{1,n}$ is precisely $\alpha_n$ in 
\eqref{eq:SA_PreFactorIID}. Furthermore, using 
$\frac{dG_n}{dz}=\frac{K}{\sqrt{n}}(z^{3}-3z)e^{-\frac{1}{2}z^{2}}$,
we obtain
\begin{multline}
    I_{2,n} = \frac{K}{\sqrt{n}}\bigg(\int_{-\frac{c_1\sqrt{n}}{\sqrt{c_2}}}^{\infty}e^{-\rhohat(z\sqrt{nc_2}+nc_1)}(z^{3}-3z)e^{-\frac{1}{2}z^{2}}dz \\
        + \int_{-\infty}^{-\frac{c_1\sqrt{n}}{\sqrt{c_2}}}e^{(1-\rhohat)(z\sqrt{nc_2}+nc_1)}(z^{3}-3z)e^{-\frac{1}{2}z^{2}}dz\bigg). \label{eq:SA_Above_17} 
\end{multline}
In accordance with the theorem statement, we must show that 
$I_{2,n}=o(\alpha_n)$ and $I_{3,n}=o(\alpha_n)$ even 
in the case that $R$ and $\rhohat$ vary with $n$.  Let $R_n = \frac{1}{n}\log M_n$
and $\rhohat_n=\rhohat(Q,R_n,s)$, and let $c_{1,n}$ and $c_{2,n}$ be the corresponding
values of $c_1$ and $c_2$.  We assume with no real loss of generality that
\begin{align}
    & \lim_{n\to\infty}R_n = R^{*} \label{eq:SA_Rstar}
\end{align}  
for some $R^{*}\ge0$ possibly equal to $\infty$.  Once the theorem
is proved for all such $R^{*}$, the same will follow for an arbitrary
sequence $\{R_n\}$. 

Table \ref{tab:SU_rcusGrowth} summarizes the growth rates $\alpha_n$, $I_{2,n}$ and
$I_{3,n}$ for various ranges of $R^*$, and indicates whether the first
or second integral (see \eqref{eq:SA_ChgMeasure5} and \eqref{eq:SA_Above_17}) dominates the behavior
of each.  We see that $I_{2,n}=o(\alpha_n)$ and $I_{3,n}=o(\alpha_n)$ for
all values of $R^{*}$, as desired.

\begin{table*}
    \caption{Growth rates of $\alpha_{n}$, $I_{2,n}$ and $I_{3,n}$ when the rate
    converges to $R^{*}$. \label{tab:SU_rcusGrowth}}
    \begin{centering}
    \begin{tabular}{|c|c|c|c|c|c|c|}
    \hline 
     & $\hat{\rho}$ & $c_{1}$ & Dominant Term(s) & $\alpha_{n}$ & $I_{2,n}$ & $I_{3,n}$\tabularnewline
    \hline 
    \hline 
    $R^{*}\in[0,\Rcrs(Q))$ & $1$ & $<0$ & 2 & $\Theta(1)$ &  $\Theta\Big(\frac{1}{\sqrt{n}}\Big)$ & $o\Big(\frac{1}{\sqrt{n}}\Big)$\tabularnewline
    \hline 
    $R^{*}=\Rcrs(Q)$ & $\to1$ & $\to0$ & 2 & $\omega\Big(\frac{1}{\sqrt{n}}\Big)$ &  $O\Big(\frac{1}{\sqrt{n}}\Big)$ & $o\Big(\frac{1}{\sqrt{n}}\Big)$\tabularnewline
    \hline 
    $R^{*}\in(\Rcrs(Q),I_{s}(Q))$ & $\in(0,1)$ & $0$ & 1,2 & $\Theta\Big(\frac{1}{\sqrt{n}}\Big)$ & $\Theta\Big(\frac{1}{n\sqrt{n}}\Big)$ & $o\Big(\frac{1}{\sqrt{n}}\Big)$\tabularnewline
    \hline 
    $R^{*}=I_{s}(Q)$ & $\to0$ & $\to0$ & 1 & $\omega\Big(\frac{1}{\sqrt{n}}\Big)$ & $O\Big(\frac{1}{\sqrt{n}}\Big)$ & $o\Big(\frac{1}{\sqrt{n}}\Big)$\tabularnewline
    \hline 
    $R^{*}>I_{s}(Q)$ & $0$ & $>0$ & 1 & $\Theta(1)$ & $\Theta\Big(\frac{1}{\sqrt{n}}\Big)$ & $o\Big(\frac{1}{\sqrt{n}}\Big)$\tabularnewline
    \hline 
    \end{tabular}
    \par\end{centering}
\end{table*}
 
The derivations of the growth rates in Table \ref{tab:SU_rcusGrowth} 
when $R^{*}\notin\{\Rcrs(Q),I_s(Q)\}$ are
done in a similar fashion to Appendix \ref{sub:SU_ASYMPTOTIC_SA}. 
To avoid repetition, we provide details only for 
$R^{*}=\Rcrs(Q)$; this is a less straightforward 
case whose analysis differs slightly from Appendix \ref{sub:SU_ASYMPTOTIC_SA}.
From Section \ref{sec:SU_REF_DEFS}, we have $\rhohat\to1$ and 
$c_{1,n}\to0$ from below, with $c_{1,n}<0$ only if $\rhohat_n=1$.

For any $\rhohat\in[0,1]$, the terms $e^{-\rhohat(\cdot)}$ and 
$e^{(1-\rhohat)(\cdot)}$ in \eqref{eq:SA_Above_17} are both
upper bounded by one across their respective ranges of integration.
Since the moments of a Gaussian random variable are finite,
it follows that both integrals are $O(1)$, and thus 
$I_{2,n}=O\big(\frac{1}{\sqrt{n}}\big)$.  The term $I_{3,n}$ is
handled similarly, so it only remains to show that
$\alpha_n=\omega\big(\frac{1}{\sqrt n}\big)$.  In the case that
$\rhohat=1$, the second integral in \eqref{eq:SA_PreFactorIID}
is at least $\frac{1}{2}$, since $c_1\le0$.  
It only remains to handle the case that
$c_1=0$ and $\rhohat_n\to1$ with $\rhohat_n<1$.  
For any $\delta>0$, we have $\rhohat_n \ge 1-\delta$ for 
sufficiently large $n$.  Lower bounding $\alpha_n$
by replacing $1-\rhohat$ by $\delta$ in the second term of 
\eqref{eq:SA_PreFactorIID}, we have similarly to \eqref{eq:SU_LatticeAsymp5} that 
\begin{equation}
    \alpha_n \ge e^{\frac{1}{2}nc_2\delta^2}\Qsf\big(\delta\sqrt{nc_2}\big)
             \asymp \frac{1}{\sqrt{2\pi nc_2}\delta}.
\end{equation}
Since $\delta$ is arbitrary, we obtain 
$\alpha_n=\omega\big(\frac{1}{\sqrt{n}}\big)$, as desired.

\subsubsection{Lattice case} \label{sub:SA_ABOVE_RCR}

The arguments following \eqref{eq:SA_Above_8} are essentially identical in the lattice case,
so we focus our attention on obtaining the analogous expression to \eqref{eq:SA_Above_8}.
Letting $P_n(z)$ denote the probability mass function (PMF) of $\sum_{i=1}^{n}Z_i$, 
we can write \eqref{eq:SA_ChgMeasure5a} as
\begin{equation}
    I_{n}= \sum_{z\ge0}P_n(z)e^{-\rhohat z} + \sum_{z<0}P_n(z)e^{(1-\rhohat)z}. \label{eq:SU_In2}
\end{equation}
Using the fact that $\EE[Z]=c_1$ and $\var[Z]=c_2>0$,
we have from the local limit theorem in \cite[Eq. (5A.12)]{Gallager} that
\begin{equation}
    P_n(z) = \phi_h(z;nc_1,nc_2) + \tilde{P}_n(z), \label{eq:SU_LatticeBE}
\end{equation}
where $\phi_h$ is defined in \eqref{eq:SU_phi_h}, and
$\tilde{P}_n(z)=o\big(\frac{1}{\sqrt{n}}\big)$ uniformly in $z$.  Thus, 
analogously to \eqref{eq:SA_Above_8}, we can write
\begin{equation}
    I_{n}=I_{1,n}+I_{2,n},\label{eq:SA_Above_8a} 
\end{equation}
where the two terms denote the right-hand side of \eqref{eq:SU_In2} with 
$\phi_h$ and $\tilde{P}_n(z)$
respectively in place of $P_n(z)$.  Using the definition of $\gamma_n$ in \eqref{eq:SU_alphan}
and the fact that $\sum_{i}Z_i$ has the same support as $nR - i_s^n(\Xv,\Yv)$ (cf. \eqref{eq:SA_Above_2}),
we see that the first summation in \eqref{eq:SU_In2} is over the 
set $\{\gamma_n + ih : i\in\ZZ,i\ge0\}$, and the second summation is over
the set $\{\gamma_n + ih : i\in\ZZ,i<0\}$.  It follows that $I_{1,n}=\alpha_n$, 
and similar arguments to the non-lattice case show that $I_{2,n}=o(\alpha_n)$.

\subsection{Proof of Theorem \ref{thm:SA_SaddleRCU}} \label{sub:SA_RCU_PROOF}

Throughout this section, we make use of the same notation as Appendix \ref{sub:SU_SADDLEPOINT_PROOF}.
We first discuss the proof of \eqref{eq:SU_SaddleRCU}.  Using the definition of
$\rcuss$, we can follow identical arguments to those following \eqref{eq:SA_Above_1} to conclude that
\begin{equation}
    \rcuss(n,M) = I_n e^{-n(\Eziid(Q,\rhohat,s)-\rhohat R)},
\end{equation}
where analogously to \eqref{eq:SA_ChgMeasure2} and \eqref{eq:SA_ChgMeasure5a}, we have
\begin{align}
    I_n &= \int_0^1\int_{\log\frac{u\sqrt{2\pi nc_3}}{\psi_s}}^\infty e^{-\rhohat z} dF_{n}(z) dF_{U}(u) \\
        &= \int_{\log \frac{\sqrt{2\pi nc_3}}{\psi_s}}^{\infty} e^{-\rhohat z} dF_n(z) \nonumber \\ 
            & \qquad\qquad + \frac{\psi_s}{\sqrt{2\pi nc_3}}\int_{-\infty}^{\log \frac{\sqrt{2\pi nc_3}}{\psi_s}} e^{(1-\rhohat)z} dF_n(z).
\end{align}
The remaining arguments in proving \eqref{eq:SU_SaddleRCU} follow
those given in Appendix \ref{sub:SU_SADDLEPOINT_PROOF}, and are omitted.  

To prove \eqref{eq:SU_PfRCU_1}, we make use of two technical lemmas, whose 
proofs are postponed until the end of the section.  The following lemma
can be considered a refinement of \cite[Lemma 47]{Finite}.

\begin{lem} \label{lem:REF_Lem20}
    Fix $K>0$, and for each $n$, let $(n_1,\cdots,n_K)$ be integers such that $\sum_{k}n_k=n$.
    Fix the PMFs $Q_1,\cdots,Q_K$ on
    a finite subset of $\RR$, and let $\sigma_1^2,\cdots,\sigma_K^2$ be the corresponding
    variances.  Let $Z_{1},\cdots,Z_{n}$ be independent random variables, $n_k$ of which
    are distributed according to $Q_k$ for each $k$.  Suppose that 
    $\min_{k}\sigma_k > 0$ and $\min_{k}n_k = \Theta(n)$.  Defining
    \begin{align}
        \Ic_0  &\triangleq \bigcup_{k \,:\, \sigma_{k} > 0}\big\{ z \,:\, Q_k(z)>0 \big\} \label{eq:SA_I0} \\
        \psi_0 &\triangleq 
            \begin{cases}
                1 & \Ic_0\text{ \em does not lie on a lattice} \\
                \frac{h_0}{1-e^{-h_0}} & \Ic_0\text{ \em lies on a lattice with span }h_0,
            \end{cases} \label{eq:SA_Psi_0}
    \end{align}
    the summation $S_n\triangleq\sum_{i}Z_i$ satisfies the 
    following uniformly in $t$:
    \begin{equation}
        \EE\Big[e^{-S_n}\emph{\openone}\big\{S_n>t\big\}\Big] \le e^{-t}\bigg(\frac{\psi_0}{\sqrt{2\pi V_n}} + o\Big(\frac{1}{\sqrt{n}}\Big)\bigg), \label{eq:SU_Lemma20}
    \end{equation}
    where $V_n \triangleq \var[S_n]$.
\end{lem}

Roughly speaking, the following lemma ensures the existence of a high 
probability set in which Lemma \ref{lem:REF_Lem20} can be applied 
to the inner probability in \eqref{eq:SU_RCU}.
We make use of the definitions in \eqref{eq:REF_Vs}--\eqref{eq:REF_sigma_s},
and we define the random variables
\begin{equation}
(\Xv,\Yv,\Xvbar,\Xv_{s}) \sim Q^{n}(\xv)W^{n}(\yv|\xv)Q^{n}(\xvbar)\Ptilde_s^{n}(\xv_{s}|\yv), \label{eq:REF_BoldVars}
\end{equation}
where $\Ptilde_s^{n}(\xv|\yv) \triangleq \prod_{i=1}^n \Ptilde_s(x_i|y_i)$.
Furthermore, we write the empirical distribution of $\yv$ as $\hat{P}_{\yv}$,
and we let $P_{\Yv}$ denote the PMF of $\Yv$.
\begin{lem} \label{lem:REF_SetF}
    Let the parameters $s>0$ and $\rhohat\in[0,1]$ be given. If the triplet $(W,q,Q)$ satisfies 
    \eqref{eq:REF_Assumption1}--\eqref{eq:REF_Assumption2}, then the set
    \begin{equation}
        \Fc_{\rhohat,s}^n(\delta) \defeq \Big\{\yv\,:\,P_{\Yv}(\yv)>0, \, \max_{y}\big|\hat{P}_{\yv}(y)-P_{\rhohat,s}^{*}(y)\big| \le \delta\Big\} \label{eq:REF_SetFn}
    \end{equation}
    satisfies the following properties:
    \begin{enumerate}
        \item For any $\yv\in\Fc_{\rhohat,s}^n(\delta)$, we have
        \begin{equation}
            \var\big[i_{s}^{n}(\Xv_{s},\Yv)\,|\,\Yv=\yv\big]\ge n(c_3 - r(\delta)), \label{eq:REF_VarLB}
        \end{equation}
        where $r(\delta)\to0$ as $\delta\to0$.
        \item For any $\delta>0$, we have
        \begin{equation}
            \hspace*{-6mm}\liminf_{n\to\infty}-\frac{1}{n}\log\frac{\sum_{\xv,\yv\notin\Fc_{\rhohat,s}^n(\delta)}Q^{n}(\xv)W^{n}(\yv|\xv)e^{-\rhohat i_{s}^{n}(\xv,\yv)}}{\sum_{\xv,\yv}Q^{n}(\xv)W^{n}(\yv|\xv)e^{-\rhohat i_{s}^{n}(\xv,\yv)}} > 0. \label{eq:REF_LimRatio}
        \end{equation}
    \end{enumerate}
\end{lem}
It should be noted that since the two statements of Lemma \ref{lem:REF_SetF} hold true
for any $\rhohat\in[0,1]$, they also hold true when $\rhohat$ varies within this range,
thus allowing us to handle rates which vary with $n$.
Before proving the lemmas, we show how they are used to obtain the desired result.

\subsubsection*{Proof of \eqref{eq:SU_PfRCU_1} based on Lemmas \ref{lem:REF_Lem20}--\ref{lem:REF_SetF}}

By upper bounding $M-1$ by $M$ and splitting $\rcu$ (see \eqref{eq:SU_RCU}) 
according to whether or not $\yv\in\Fc_{\rhohat,s}^n(\delta)$, we obtain
\begin{align}
    & \rcu(n,M) \le \sum_{\xv,\yv\in\Fc_{\rhohat,s}^n(\delta)}Q^n(\xv)W^{n}(\yv|\xv) \nonumber \\
        & \qquad\qquad\qquad \times \min\Big\{1,M\PP[i_{s}^{n}(\Xvbar,\yv)\ge i_{s}^{n}(\xv,\yv)]\Big\} \nonumber \\
        & + M^{\rhohat}\sum_{\xv,\yv\notin\Fc_{\rhohat,s}^n(\delta)}Q^{n}(\xv)W^{n}(\yv|\xv)e^{-\rhohat i_{s}^{n}(\xv,\yv)}, \label{eq:REF_ExpandedRCU}
\end{align} 
where we have replaced $q^n$ by $i_s^n$ since each is a monotonically increasing
function of the other, and in the summation over $\yv\notin\Fc_{\rhohat,s}^n(\delta)$ we 
further weakened the bound using Markov's inequality and $\min\{1,\cdot\} \le (\cdot)^{\rhohat}$.
In order to make the inner probability in \eqref{eq:REF_ExpandedRCU}
more amenable to an application of Lemma \ref{lem:REF_Lem20}, we
follow \cite[Sec. 3.4.5]{FiniteThesis} and note that the following holds
whenever $\Ptilde_s^{n}(\xvbar|\yv) > 0$:
\begin{align}
    Q^{n}(\xvbar) & =Q^{n}(\xvbar)\frac{\Ptilde_s^{n}(\xvbar|\yv)}{\Ptilde_s^{n}(\xvbar|\yv)} = \Ptilde_s^{n}(\xvbar|\yv)e^{-i_{s}^{n}(\xvbar,\yv)}.\label{eq:iidDeriv0}
\end{align}
For a fixed sequence $\yv$ and a constant $t$, summing \eqref{eq:iidDeriv0} over
all $\xvbar$ such that $i_{s}^{n}(\xvbar,\yv)\ge t$ yields
\begin{equation}
    \mathbb{P}[i_{s}^{n}(\Xvbar,\yv)\ge t]
    =\mathbb{E}\Big[e^{-i_{s}^{n}(\Xv_{s},\Yv)}\openone\big\{ i_{s}^{n}(\Xv_{s},\Yv)\ge t\big\}\,\Big|\,\Yv=\yv\Big]\label{eq:iidDeriv1}
\end{equation}
under the joint distribution in \eqref{eq:REF_BoldVars}. 

We now observe that \eqref{eq:iidDeriv1} is of the same form as
the left-hand side of \eqref{eq:SU_Lemma20}.  We apply Lemma 
\ref{lem:REF_Lem20} with $Q_k$ given by the PMF of $i_s(X_s,y_k)$
under $X_s \sim \Ptilde_s(\,\cdot\,|y_k)$, where $y_k$ is the $k$-th output symbol  
for which $\sum_{x}Q(x)W(y|x)>0$. The conditions of the lemma are easily seen to be satisfied
for sufficiently small $\delta$ due to the definition of
$\Fc_{\rhohat,s}^n(\delta)$ in \eqref{eq:REF_SetFn}, the
assumption in \eqref{eq:REF_Assumption2}, and \eqref{eq:REF_VarProof4}.
We have from \eqref{eq:SU_Lemma20}, \eqref{eq:REF_VarLB} and \eqref{eq:iidDeriv1} that
\begin{equation}
    \PP\big[i_{s}^{n}(\Xvbar,\yv)\ge t\big] \le \frac{\psi_s}{\sqrt{2\pi n(c_3-r(\delta))}}e^{-t}(1+o(1)) \label{eq:REF_Tail1}
\end{equation}
for all $\yv\in\Fc_{\rhohat,s}^n(\delta)$ and sufficiently small $\delta$.
Here we have used the fact that $\psi_0$ in \eqref{eq:SA_Psi_0}
coincides with $\psi_s$ in \eqref{eq:SA_Psi_s}, which follows 
from \eqref{eq:REF_VarProof4} and the fact that $\Ptilde_s(x|y)>0$ if and
only if $Q(x)W(y|x)>0$ (see \eqref{eq:REF_Assumption1} and \eqref{eq:REF_Vs}).

Using the uniformity of the $o(1)$ term in $t$ in \eqref{eq:REF_Tail1} 
(see Lemma \ref{lem:REF_Lem20}), taking $\delta\to0$ (and 
hence $r(\delta)\to0$), and writing
\begin{equation}
    \min\{1,f_n(1+\zeta_n)\} \le (1 + |\zeta_n|)\min\{1,f_n\} \label{eq:SA_min1property},
\end{equation}  
we see that the first term in \eqref{eq:REF_ExpandedRCU} 
is upper bounded by $\rcuss(n,M)(1+o(1))$.  To complete the proof of \eqref{eq:SU_PfRCU_1}, 
we must show that the second term in \eqref{eq:REF_ExpandedRCU} can be incorporated
into the multiplicative $1+o(1)$ term.  To see this, we note from \eqref{eq:SA_RCU_hat} and \eqref{eq:SU_SaddleRCU}
that the exponent of $\rcuss$ is given by $\Eziid(Q,\rhohat,s)-\rhohat R$.
From \eqref{eq:SU_E0s_IID}, the denominator in the logarithm in 
\eqref{eq:REF_LimRatio} equals $e^{-n\Eziid(Q,\rhohat,s)}$.  Combining
these observations, the second part of Lemma \ref{lem:REF_SetF} shows
that the second term in \eqref{eq:REF_ExpandedRCU} decays at a faster 
exponential rate than $\rcuss$, thus yielding the desired result.
 
\subsubsection*{Proof of Lemma \ref{lem:REF_Lem20}}

The proof makes use of the local limit theorems given in \cite[Thm. 1]{LLTNonLattice}
and \cite[Sec. VII.1, Thm. 2]{PetrovBook} for the non-lattice and lattice 
cases respectively. 
We first consider the summation $S'_n \triangleq \sum_{i=1}^{n'}Z_i$, where we 
assume without loss of generality that the first $n'=\Theta(n)$ indices
correspond to positive variances, and the remaining $n-n'$ correspond to
zero variances.  We similarly assume that $\sigma_{k}>0$ for $k=1,\cdots,K'$,
and $\sigma_{k}=0$ for $k=K'+1,\cdots,K$. We clearly have $\var[S'_n]=\var[S_n]=V_n$.

We first consider the non-lattice case.  We claim that the conditions of the
lemma imply the following local limit theorem given in \cite[Thm. 1]{LLTNonLattice}:
\begin{equation}
    \PP\big[S'_n \in [z,z+\eta)\big] = \frac{\eta}{\sqrt{2\pi V_n}}e^{-\frac{(z - \mu'_n)^2}{2V_n}} + o\Big(\frac{1}{\sqrt{n}}\Big) \label{eq:REF_LLTNonLattice}
\end{equation}
uniformly in $z$, where $\mu'_n\triangleq\EE[S'_n]$, and $\eta>0$ is arbitrary.
To show this, we must verify the technical assumptions given in \cite[p. 593]{LLTNonLattice}.
First, \cite[Cond. ($\alpha$)]{LLTNonLattice} states that there exists $Z_{\mathrm{max}}<\infty$ 
and $c>0$ such that 
\begin{equation}
\frac{1}{\var[Z]}\EE\big[(Z-\EE[Z])^2 \openone\{|Z-\EE[Z]| \le Z_{\mathrm{max}}\}\big]>c
\end{equation}
under $Z \sim Q_k$ and each $k=1,\cdots,K'$.  This is trivially satisfied since we are 
considering finite alphabets, which implies that the support of each $Q_k$ is bounded. 
The Lindeberg condition is stated in
\cite[Cond. ($\gamma$)]{LLTNonLattice}, and is trivially satisfied due to the assumption that
$n_k=\Theta(n)$ for all $k$. The only non-trivial condition is 
\cite[Cond. ($\beta$)]{LLTNonLattice}, which can be written as follows in the
case of finite alphabets:  For any given lattice, there exists $\delta>0$ such that
\begin{align}
\hspace*{-2.5mm}\frac{1}{\log V_n}\sum_{i=1}^{n'}\PP[Z_i\text{ is not }\delta\text{-close to a lattice point}]\to\infty.
\end{align}    
Since we are considering the case that $\Ic_0$ does not lie on a lattice,
we have for sufficiently small $\delta$ that the summation grows linearly in $n$,
whereas $\log V_n$ only grows as $\log n$.  We have thus shown that each of the
technical conditions in \cite{LLTNonLattice} is satisfied, and hence 
\eqref{eq:REF_LLTNonLattice} holds. 

Upper bounding the exponential term in \eqref{eq:REF_LLTNonLattice} by one,
and noting that $S_n-S'_n$ has zero variance, we obtain
\begin{equation}
    \PP\big[S_n \in [z,z+\eta)\big] \le \frac{\eta}{\sqrt{2\pi V_n}} + o\Big(\frac{1}{\sqrt{n}}\Big) \label{eq:REF_LLTNonLattice2}
\end{equation}
uniformly in $z$.  We can now prove the lemma 
similarly to \cite[Lemma 47]{Finite} by writing
\begin{align}
    & \EE\Big[e^{-S_n}\openone\big\{S_n>t\big\}\Big] \nonumber \\ 
        & \qquad \le \sum_{l=0}^{\infty} e^{-t-l\eta}\PP\Big[t+l\eta \le S_n \le t+(l+1)\eta \Big] \\
        & \qquad \le \sum_{l=0}^{\infty} e^{-t-l\eta}\bigg( \frac{\eta}{\sqrt{2\pi V_n}} + o\Big(\frac{1}{\sqrt{n}}\Big) \bigg) \\
        & \qquad = e^{-t}\bigg( \frac{\eta}{(1-e^{-\eta})\sqrt{2\pi V_n}} + o\Big(\frac{1}{\sqrt{n}}\Big) \bigg), \label{eq:SA_Lem20Step3} 
\end{align}
where \eqref{eq:SA_Lem20Step3} follows by evaluating the summation using
the geometric series.  The proof is concluded by taking $\eta\to0$ and
using the identity $\lim_{\eta\to0}\frac{\eta}{1-e^{-\eta}}=1$.
The uniformity of \eqref{eq:SU_Lemma20} in $t$ follows from the 
uniformity of \eqref{eq:REF_LLTNonLattice2} in $z$.

In the lattice case, the argument is essentially unchanged, but we instead use
the local limit theorem given in \cite[Sec. VII.1, Thm. 2]{PetrovBook}, which yields
\begin{equation}
    \PP[S'_n = z] = \frac{h_0}{\sqrt{2\pi V_n}}e^{-\frac{(z - \mu'_n)^2}{2V_n}} + o\Big(\frac{1}{\sqrt{n}}\Big) 
\end{equation} 
uniformly in $z$ on the lattice corresponding to $S'_n$ (with span $h_0$).
The remaining arguments are identical to the non-lattice case, with $\eta=h_0$
instead of $\eta\to0$.

\subsubsection*{Proof of Lemma \ref{lem:REF_SetF}}

We obtain \eqref{eq:REF_VarLB} by using the definitions of 
$c_3$ and  $\Fc_{\rhohat,s}^n(\delta)$ 
(see \eqref{eq:REF_sigma_s} and \eqref{eq:REF_SetFn}) to write
\begin{align}
    & \var[i_{s}^{n}(\Xv_{s},\Yv)\,|\,\Yv=\yv] \nonumber \\ 
        & \qquad =\sum_y n\hat{P}_{\yv}(y) \var[i_{s}^{n}(X_{s},Y)\,|\, Y=y]\\
        & \qquad \ge\sum_{y}n(P^{*}_{\rhohat,s}(y)-\delta)\var[i_{s}^{n}(X_{s},Y)\,|\, Y=y] \\
        & \qquad = n\big(c_3 - o(\delta)\big),
\end{align}
where $(X_s|Y=y) \sim \Ptilde_s(\,\cdot\,|\,y)$.
To prove the second property, we perform an expansion in terms of types in the
same way as Appendix \ref{sub:SU_PRIMAL_PROOF} to conclude that the exponent 
of the denominator in the logarithm in \eqref{eq:REF_LimRatio} is given by 
\begin{equation}
    \min_{P_{XY}}\sum_{x,y}P_{XY}(x,y)\log\Bigg(\frac{P_{XY}(x,y)}{Q(x)W(y|x)}e^{\rhohat i_s(x,y)}\Bigg).\label{eq:LessConstr}
\end{equation}
Similarly, using the definition of $\Fc_{\rhohat,s}^n(\delta)$ in \eqref{eq:REF_SetFn}, 
the exponent of the numerator in the logarithm in \eqref{eq:REF_LimRatio} is given by
\begin{multline}
    \min_{P_{XY} \,:\, \max_{y}|P_Y(y)-P_{\rhohat,s}^{*}(y)|>\delta} \\ \sum_{x,y}P_{XY}(x,y)\log\Bigg(\frac{P_{XY}(x,y)}{Q(x)W(y|x)}e^{\rhohat i_s(x,y)}\Bigg).\label{eq:MoreConstr}
\end{multline}
A straightforward evaluation of the KKT conditions
\cite[Sec. 5.5.3]{Convex} yields that \eqref{eq:LessConstr} is uniquely
minimized by $P^{*}_{\rhohat,s}(x,y)$, defined in \eqref{eq:SU_P*XY}.  On the
other hand, $P^{*}_{\rhohat,s}(x,y)$ does not satisfy the constraint in \eqref{eq:MoreConstr},
and thus \eqref{eq:MoreConstr} is strictly greater than \eqref{eq:LessConstr}.
This concludes the proof of \eqref{eq:REF_LimRatio}.

\bibliographystyle{IEEEtran}
\bibliography{12-Paper,18-MultiUser,18-SingleUser,35-Other}

\begin{thebibliography}{10}
\providecommand{\url}[1]{#1}
\csname url@samestyle\endcsname
\providecommand{\newblock}{\relax}
\providecommand{\bibinfo}[2]{#2}
\providecommand{\BIBentrySTDinterwordspacing}{\spaceskip=0pt\relax}
\providecommand{\BIBentryALTinterwordstretchfactor}{4}
\providecommand{\BIBentryALTinterwordspacing}{\spaceskip=\fontdimen2\font plus
\BIBentryALTinterwordstretchfactor\fontdimen3\font minus
  \fontdimen4\font\relax}
\providecommand{\BIBforeignlanguage}[2]{{%
\expandafter\ifx\csname l@#1\endcsname\relax
\typeout{** WARNING: IEEEtran.bst: No hyphenation pattern has been}%
\typeout{** loaded for the language `#1'. Using the pattern for}%
\typeout{** the default language instead.}%
\else
\language=\csname l@#1\endcsname
\fi
#2}}
\providecommand{\BIBdecl}{\relax}
\BIBdecl

\bibitem{Csiszar1}
I.~Csisz\'{a}r and J.~K\"{o}rner, ``Graph decomposition: A new key to coding
  theorems,'' \emph{IEEE Trans. Inf. Theory}, vol.~27, no.~1, pp. 5--12, Jan.
  1981.

\bibitem{Hui}
J.~Hui, ``Fundamental issues of multiple accessing,'' Ph.D. dissertation, MIT,
  1983.

\bibitem{Compound}
G.~Kaplan and S.~Shamai, ``Information rates and error exponents of compound
  channels with application to antipodal signaling in a fading environment,''
  \emph{Arch. Elek. \"{U}ber.}, vol.~47, no.~4, pp. 228--239, 1993.

\bibitem{Csiszar2}
I.~Csisz\'{a}r and P.~Narayan, ``Channel capacity for a given decoding
  metric,'' \emph{IEEE Trans. Inf. Theory}, vol.~45, no.~1, pp. 35--43, Jan.
  1995.

\bibitem{Merhav}
N.~Merhav, G.~Kaplan, A.~Lapidoth, and S.~Shamai, ``On information rates for
  mismatched decoders,'' \emph{IEEE Trans. Inf. Theory}, vol.~40, no.~6, pp.
  1953--1967, Nov. 1994.

\bibitem{ConverseMM}
V.~Balakirsky, ``A converse coding theorem for mismatched decoding at the
  output of binary-input memoryless channels,'' \emph{IEEE Trans. Inf. Theory},
  vol.~41, no.~6, pp. 1889--1902, Nov. 1995.

\bibitem{MMRevisited}
A.~Ganti, A.~Lapidoth, and E.~Telatar, ``Mismatched decoding revisited:
  {G}eneral alphabets, channels with memory, and the wide-band limit,''
  \emph{IEEE Trans. Inf. Theory}, vol.~46, no.~7, pp. 2315--2328, Nov. 2000.

\bibitem{MacMM}
A.~Lapidoth, ``Mismatched decoding and the multiple-access channel,''
  \emph{IEEE Trans. Inf. Theory}, vol.~42, no.~5, pp. 1439--1452, Sept. 1996.

\bibitem{Gallager}
R.~Gallager, \emph{Information Theory and Reliable Communication}.\hskip 1em
  plus 0.5em minus 0.4em\relax John Wiley \& Sons, 1968.

\bibitem{Strassen}
V.~Strassen, ``Asymptotische {A}bschätzungen in {S}hannon's
  {I}nformationstheorie,'' in \emph{Trans. 3rd Prague Conf. on Inf. Theory},
  1962, pp. 689--723, [{E}nglish {T}ranslation:
  http://www.math.wustl.edu/{\string~}luthy/strassen.pdf].

\bibitem{Finite}
Y.~Polyanskiy, V.~Poor, and S.~Verd\'{u}, ``Channel coding rate in the finite
  blocklength regime,'' \emph{IEEE Trans. Inf. Theory}, vol.~56, no.~5, pp.
  2307--2359, May 2010.

\bibitem{Hayashi}
M.~Hayashi, ``Information spectrum approach to second-order coding rate in
  channel coding,'' \emph{IEEE Trans. Inf. Theory}, vol.~55, no.~11, pp.
  4947--4966, Nov. 2009.

\bibitem{SaddlepointBook}
J.~L. Jensen, \emph{Saddlepoint Approximations}.\hskip 1em plus 0.5em minus
  0.4em\relax Oxford University Press, 1995.

\bibitem{CsiszarBook}
I.~Csisz\'{a}r and J.~K\"{o}rner, \emph{Information Theory: Coding Theorems for
  Discrete Memoryless Systems}, 2nd~ed.\hskip 1em plus 0.5em minus 0.4em\relax
  Cambridge University Press, 2011.

\bibitem{GallagerCC}
R.~Gallager, ``Fixed composition arguments and lower bounds to error
  probability,'' \url{http://web.mit.edu/gallager/www/notes/notes5.pdf}.

\bibitem{Variations}
S.~Shamai and I.~Sason, ``Variations on the {G}allager bounds, connections, and
  applications,'' \emph{IEEE Trans. Inf. Theory}, vol.~48, no.~12, pp.
  3029--3051, Dec. 2002.

\bibitem{Convex}
S.~Boyd and L.~Vandenberghe, \emph{Convex Optimization}.\hskip 1em plus 0.5em
  minus 0.4em\relax Cambridge University Press, 2004.

\bibitem{MMSomekh}
A.~Somekh-Baruch, ``On achievable rates for channels with mismatched
  decoding,'' 2013, submitted to \emph{IEEE Trans. Inf. Theory} [Online:
  http://arxiv.org/abs/1305.0547].

\bibitem{JournalMU}
J.~Scarlett, A.~Martinez, and A.~{Guill{\'e}n i F\`{a}bregas}, ``Multiuser
  coding techniques for mismatched decoding,'' 2013, submitted to \emph{IEEE
  Trans. Inf. Theory} [Online: http://arxiv.org/abs/1311.6635].

\bibitem{PaperExpurg}
J.~Scarlett, L.~Peng, N.~Merhav, A.~Martinez, and A.~{Guill\'{e}n i
  F\`{a}bregas}, ``Expurgated random-coding ensembles: Exponents, refinements
  and connections,'' 2013, submitted to \emph{IEEE Trans. Inf. Theory} [Online:
  http://arxiv.org/abs/1307.6679].

\bibitem{MMGeneralFormula}
A.~Somekh-Baruch, ``A general formula for the mismatch capacity,''
  http://arxiv.org/abs/1309.7964.

\bibitem{RefinementJournal}
Y.~Altu\u{g} and A.~B. Wagner, ``Refinement of the random coding bound,'' 2014,
  http://arxiv.org/abs/1312.6875.

\bibitem{TwoChannels}
P.~Elias, ``Coding for two noisy channels,'' in \emph{Third London Symp. Inf.
  Theory}, 1955.

\bibitem{ShulmanThesis}
N.~Shulman, ``Communication over an unknown channel via common broadcasting,''
  Ph.D. dissertation, Tel Aviv University, 2003.

\bibitem{Stone}
C.~Stone, ``On local and ratio limit theorems,'' in \emph{Proc. Fifth Berkeley
  Symp. Math. Stat. Prob.}, 1965, pp. 217--224.

\bibitem{FanoBook}
R.~Fano, \emph{Transmission of information: A statistical theory of
  communications}.\hskip 1em plus 0.5em minus 0.4em\relax MIT Press, 1961.

\bibitem{TightAverage}
R.~Gallager, ``The random coding bound is tight for the average code,''
  \emph{IEEE Trans. Inf. Theory}, vol.~19, no.~2, pp. 244--246, March 1973.

\bibitem{DyachkovCC}
A.~G. D'yachkov, ``Bounds on the average error probability for a code ensemble
  with fixed composition,'' \emph{Prob. Inf. Transm.}, vol.~16, no.~4, pp.
  3--8, 1980.

\bibitem{YALMIP}
J.~L\"{o}fberg, ``{YALMIP} : A toolbox for modeling and optimization in
  {MATLAB},'' in \emph{Proc. CACSD Conf.}, Taipei, Taiwan, 2004.

\bibitem{PaperITA}
J.~Scarlett, A.~Martinez, and A.~{Guill\'{e}n i F\`{a}bregas},
  ``Cost-constrained random coding and applications,'' in \emph{Inf. Theory and
  Apps. Workshop}, San Diego, CA, Feb. 2013.

\bibitem{ModDevMaths}
H.~Rubin and J.~Sethuraman, ``Probabilities of moderate deviations,''
  \emph{Indian Journal of Stats.}, vol.~27, no.~2, pp. 325--346, Dec. 1965.

\bibitem{Feller}
W.~Feller, \emph{An introduction to probability theory and its applications},
  2nd~ed.\hskip 1em plus 0.5em minus 0.4em\relax John Wiley \& Sons, 1971,
  vol.~2.

\bibitem{FiniteThesis}
Y.~Polyanskiy, ``Channel coding: Non-asymptotic fundamental limits,'' Ph.D.
  dissertation, Princeton University, 2010.

\bibitem{Dobrushin}
R.~L. Dobrushin, ``Asymptotic estimates of the probability of error for
  transmission of messages over a discrete memoryless communication channel
  with a symmetric transition probability matrix,'' \emph{Theory Prob. Apps.},
  vol.~7, no.~e, pp. 270--300, 1962.

\bibitem{RefinementSP}
Y.~Altu\u{g} and A.~B. Wagner, ``Refinement of the sphere-packing bound:
  Asymmetric channels,'' 2012, http://arxiv.org/abs/1211.6697.

\bibitem{ModerateDev}
------, ``Moderate deviations in channel coding,'' 2012,
  http://arxiv.org/abs/1208.1924.

\bibitem{ModerateDev2}
Y.~Polyanskiy and S.~Verdu, ``Channel dispersion and moderate deviations limits
  for memoryless channels,'' in \emph{Allerton Conf. on Comms., Control and
  Comp.}, 2010.

\bibitem{Saddlepoint}
A.~Martinez and A.~{Guill\'{e}n i F\`{a}bregas}, ``Saddlepoint approximation of
  random-coding bounds,'' in \emph{Inf. Theory App. Workshop}, La Jolla, CA,
  2011.

\bibitem{BahadurRao}
R.~Bahadur and R.~{Ranga Rao}, ``On deviations of the sample mean,''
  \emph{Annals Math. Stats.}, vol.~31, pp. 1015--1027, Dec. 1960.

\bibitem{PaperRefinement}
J.~Scarlett, A.~Martinez, and A.~{Guill{\'e}n i F\`{a}bregas}, ``A derivation
  of the asymptotic random-coding prefactor,'' in \emph{Allerton Conf. on
  Comm., Control and Comp.}, Monticello, IL, 2013.

\bibitem{Minimax}
K.~Fan, ``Minimax theorems,'' \emph{Proc. Nat. Acad. Sci.}, vol.~39, pp.
  42--47, 1953.

\bibitem{LLTNonLattice}
J.~Mineka and S.~Silverman, ``A local limit theorem and recurrence conditions
  for sums of independent non-lattice random variables,'' \emph{Annals Math.
  Stats.}, vol.~41, no.~2, pp. 592--600, April 1970.

\bibitem{PetrovBook}
V.~V. Petrov, \emph{Sums of Independent Random Variables}.\hskip 1em plus 0.5em
  minus 0.4em\relax Springer-Verlag, 1975.

\end{thebibliography}

\end{document}